\documentclass[twoside,11pt]{article}

\usepackage{blindtext}

\usepackage[preprint]{jmlr2e}
\usepackage[dvipsnames]{xcolor}
\usepackage{amsmath, dsfont}

\makeatletter
\renewenvironment{proof}[1][Proof]{%
  \par\noindent{\bf #1\ }%
}{\hfill\BlackBox\\[2mm]}
\makeatother

\usepackage[utf8]{inputenc}
\usepackage{tikz}
\usetikzlibrary{positioning, calc, fit, decorations.pathreplacing}

\newcommand{\RR}{\mathbb{R}}
\newcommand{\NN}{\mathbb{N}}
\newcommand{\EP}[1]{\mathbb{E}_P\left[ #1 \right]}
\newcommand{\VP}[1]{\mathrm{Var}_P\left[ #1 \right]}
\newcommand{\PP}[1]{\mathbb{P}_P\left( #1 \right)}
\newcommand{\EPepsilon}[1]{\mathbb{E}_{P_\epsilon}\left[ #1 \right]}

\newcommand{\dd}{\mathrm{d}}
\newcommand{\norm}[1]{\left\lVert#1\right\rVert}
\renewcommand{\epsilon}{\varepsilon} 

\DeclareMathOperator*{\argmin}{argmin}

\usepackage{enumitem}
\usepackage{subcaption}

\newtheorem{example}{Example}

\let\oldexample\example
\let\oldendexample\endexample

\renewenvironment{example}[1][]{%
  \if\relax\detokenize{#1}\relax
    \oldexample
  \else
    \oldexample[#1]%
  \fi
  \normalfont
}{%
  \hfill$\circ$\oldendexample
}

\usepackage{tikz}
\usetikzlibrary{positioning,arrows.meta,calc}
\usepackage{float}

\usepackage{algorithm}
\usepackage{algpseudocode}

\newcommand{\RN}[1]{%
  (\textup{\uppercase\expandafter{\romannumeral#1}})%
}

\newcommand{\ind}{\mathds{1}}

\usepackage{lastpage}
\jmlrheading{23}{2022}{1-\pageref{LastPage}}{1/21; Revised 5/22}{9/22}{21-0000}{
Asger Waagepetersen, Asbjørn Risom, Niels Richard Hansen, and Anton Rask Lundborg }

\ShortHeadings{Outcome-adapted AutoDML}{Waagepetersen, Risom, Hansen, and Lundborg}
\firstpageno{1}

\begin{document}

\newpage
\title{Outcome-adapted Automatic Debiased Machine Learning}

\author{\name Asger Waagepetersen$^*$ \email aw@math.ku.dk\\
   \name Asbjørn Risom$^*$ \email ar@math.ku.dk \\
   \name Niels Richard Hansen \email niels.r.hansen@math.ku.dk \\
   \name Anton Rask Lundborg \email arl@math.ku.dk \\
   \addr Department of Mathematical Sciences\\
   University of Copenhagen\\
   Universitetsparken 5, 2100 København Ø, Denmark \\
   $^*$The first two authors contributed equally to the paper.
}

\editor{My editor}

\maketitle

\begin{abstract}
Parameters of interest in causal inference, such as treatment or policy effects, can often be expressed as linear functionals of an outcome regression function. Automatic debiased machine learning (AutoDML) is a unified framework for obtaining asymptotically normal estimators of such parameters, which requires estimation of both a regression function and a Riesz representer. Existing AutoDML neural network architectures, such as RieszNet and MADNet, use a shared intermediate covariate representation. However, it remains unclear whether this shared representation should be predictive of the Riesz representer or the outcome.

We show that a shared representation of the covariates that preserves predictive power of the outcome while discarding information about the Riesz representer is asymptotically more efficient than the baseline AutoDML estimator that uses all covariates. Motivated by these results, we propose the outcome-adapted AutoDML estimator and establish its asymptotic behavior in a sample splitting framework. We provide a neural network implementation of the estimator that learns a sparse representation of the covariates that is predictive of the outcome but not predictive of the Riesz representer. We demonstrate the efficiency gains of our estimator over existing alternatives on synthetic data and achieve state-of-the-art estimation accuracy on the semi-synthetic IHDP benchmark dataset. 
\end{abstract}


\section{Introduction}
A number of important parameters in causal inference and econometrics, such as treatment and policy effects, depend on the regression function $\gamma(x) = \mathbb{E}[Y \mid X = x]$ of an outcome $Y$ given a set of covariates $X$. These parameters can  be estimated naively in terms of an estimator of the regression function, but it is well known that such plug-in estimators can suffer from considerable bias and fail to achieve $\sqrt{n}$-consistency \citep{chernozhukov2018doubledebiasedmachinelearningtreatment}. To solve this problem, \citet{chernozhukov2018doubledebiasedmachinelearningtreatment} introduced double/debiased machine learning (DML) based on classical semiparametric estimation theory, which gives asymptotically normal estimators of the parameters of interest while allowing for regression function estimators based on flexible machine learning methods.

\citet{Chernozhukov_2022_auto_dml_lasso} subsequently introduced automatic debiased machine learning (AutoDML), which identifies a common structure across a range of specific examples leading to a problem independent, hence automatic, debiasing step. They achieved this by showing that many parameters of interest are given in terms of a continuous linear functional on an $L^2$-space evaluated in the regression function. This gives an inner product representation of the functional through its Riesz representer $\alpha \in L^2$, which can be used to construct an asymptotically normal estimator using a doubly robust estimating function. \citet{chernozhukov_24_via_riesz_regression} additionally showed that the Riesz representer can be estimated without knowing its functional form through a procedure termed Riesz regression. 

AutoDML is a powerful and general framework for achieving asymptotically normal estimation of parameters of interest by leveraging modern machine learning methods for estimation of the two nuisance functions, the regression function and the Riesz representer. An example of such an AutoDML estimator is based on the RieszNet neural network architecture by \citet{chernozhukov2022riesznetforestrieszautomaticdebiased}, see Figure \ref{fig:riesz_net_architecture}, which estimates the regression function and the Riesz representer simultaneously by learning a representation $Z$ of the covariates that is shared by both nuisance functions. An outstanding question when training RieszNet, or any other AutoDML machine learning architecture, is how to optimally combine the estimation of the two nuisance functions to achieve the most efficient estimator of the parameter of interest.

\begin{figure}
    \centering
    \begin{tikzpicture}[
        scale=0.9, 
        node distance=1.5cm and 1.2cm,
        layer/.style={rectangle, draw, fill=blue!10, minimum width=2.2cm, minimum height=0.9cm, rounded corners, align=center, font=\small},
        branch_layer/.style={rectangle, draw, minimum width=2.2cm, minimum height=0.8cm, rounded corners, align=center, font=\small},
        input/.style={circle, draw, fill=green!10, minimum size=0.8cm, font=\small},
        output/.style={circle, draw, fill=red!10, minimum size=0.8cm, font=\small},
        rep_node/.style={circle, draw, fill=blue!10, minimum size=0.8cm, font=\small},
        arrow/.style={-stealth, thick},
        label_text/.style={font=\bfseries\small}
    ]

    \node[input] (W) at (0,0) {$X$};
    \node[layer, right=1.2cm of W] (T1) {Hidden \\ Layers};
    \node[rep_node, right=1.2cm of T1] (T2) {$Z$};
    \path (T1.north) -- (T2.north) node[label_text, midway, above=0.4cm] {Shared Trunk};
    \node[branch_layer, fill=purple!5, above right=0.3cm and 1.2cm of T2] (G1) {Hidden \\ Layers};
    \node[output, right=1.2cm of G1] (gamma) {$\gamma$};
    \node[label_text, purple, above=0.2cm of G1] {Outcome Branch};
    \node[branch_layer, fill=orange!5, below right=0.3cm and 1.2cm of T2] (A1) {Hidden \\ Layers};
    \node[output, right=1.2cm of A1] (alpha) {$\alpha$};
    \node[label_text, orange, below=0.2cm of A1] {Riesz Branch};
    \draw[arrow] (W) -- (T1);
    \draw[arrow] (T1) -- (T2);
    \coordinate (split) at ($(T2.east) + (0.5,0)$);
    \draw[thick] (T2.east) -- (split);
    \draw[arrow] (split) |- (G1.west);
    \draw[arrow] (split) |- (A1.west);
    
    \draw[arrow] (G1) -- (gamma);
    \draw[arrow] (A1) -- (alpha);

    \end{tikzpicture}
    \caption{RieszNet neural network architecture mapping covariates $X$ to a regression function, $\gamma$, and a Riesz representer, $\alpha$, through a shared state $Z$. The implementation used by \citet{chernozhukov2022riesznetforestrieszautomaticdebiased} uses no hidden layers in the Riesz branch and outputs the Riesz representer as a linear function of $Z$.}
    \label{fig:riesz_net_architecture}
\end{figure}
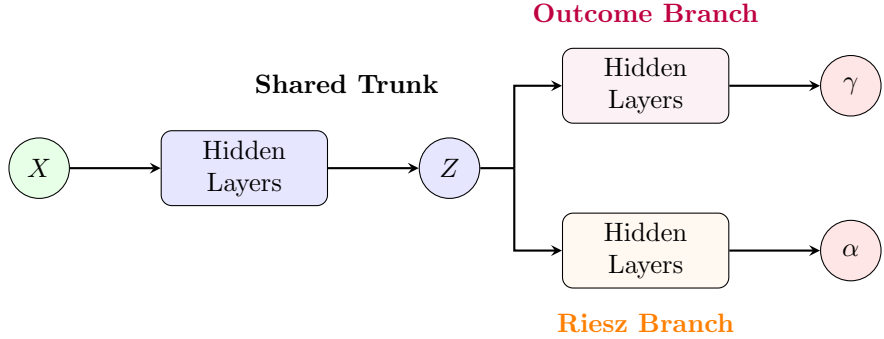

\subsection{Novel Contributions}

To optimize estimation efficiency, we introduce the outcome-adapted AutoDML estimator in Section \ref{sec:OA_Autodml}, which learns a covariate representation $Z$ exclusively through an outcome regression step. This estimator is supported by our two main theoretical results, Theorems \ref{thm:delete_overadjust} and \ref{thm:supp_with_precision}, on using a shared covariate representation $Z$ for both the regression function and the Riesz representer. Their combined interpretation, in terms of improving the efficiency of the resulting AutoDML estimator, is that $Z$ should be maximally predictive of the outcome but minimally predictive of the Riesz representer.

We propose in Section \ref{sec:neural_nets} an implementation of the outcome-adapted AutoDML estimator using a neural network with an architecture similar to the RieszNet architecture in Figure \ref{fig:riesz_net_architecture}. The fundamental novelty of our implementation is that parameters in the shared trunk are learned only in combination with learning the outcome branch, while we freeze those parameters when learning the Riesz branch. 
To maximize the efficiency gain of our estimator, we also introduce an information bottleneck based on dimensionality reduction to make the representation minimally predictive of the Riesz representer.

We demonstrate the performance of our outcome-adapted AutoDML estimator on a 
standard semi-synthetic Infant Health and Development Program (IHDP) benchmark data\-set \citep{chernozhukov2022riesznetforestrieszautomaticdebiased}, where we achieve state-of-the-art performance in comparison to other doubly robust estimators, see Figure \ref{fig:mae}. The details of this experiment are given in Section \ref{sec:IHDP}. Section \ref{sec:sim} provides additional numerical experiments that also demonstrate the benefits of the outcome-adapted AutoDML estimator over alternative ways to trade off different objectives when learning the nuisance functions.

\begin{figure}
    \centering
    \includegraphics[width=\linewidth]{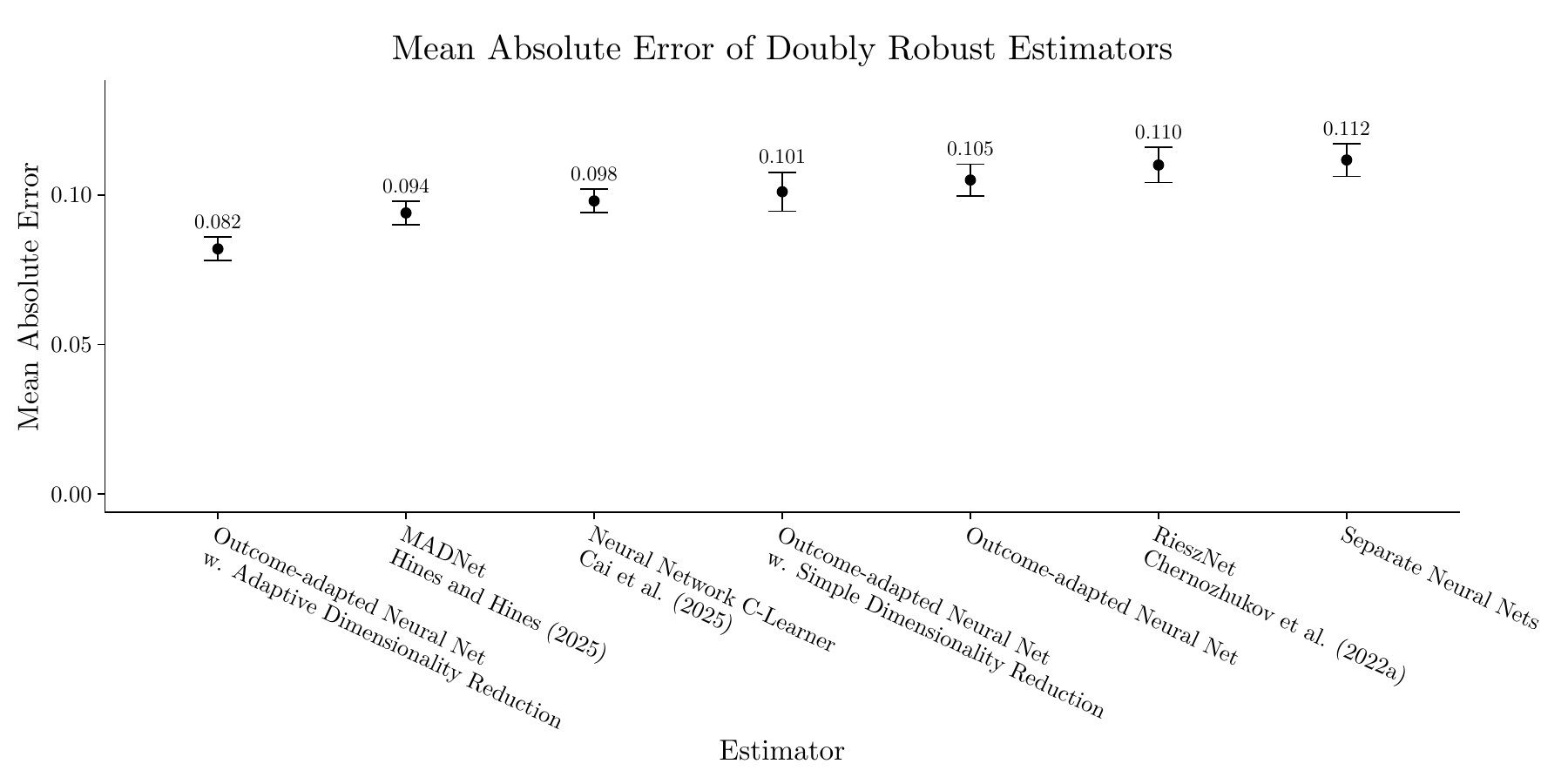}
    \caption{Mean absolute error (MAE) of average treatment effect estimates on the IHDP dataset using different doubly robust estimators. Error bars indicate asymptotic 95\% Gaussian confidence intervals.}
    \label{fig:mae}
\end{figure}

\subsection{Relations to the Literature}

Our work takes as a starting point the recent AutoDML papers \citep{Chernozhukov_2022_auto_dml_lasso, chernozhukov_24_via_riesz_regression}, the RieszNet implementation of AutoDML \citep{chernozhukov2022riesznetforestrieszautomaticdebiased}, and prior neural network DML estimators, such as Dragonnet \citep{shi2019adapting}.
This line of work advocates that the neural networks should be designed and trained to make the representation $Z$ predictive of the Riesz representer rather than the outcome. In practice, RieszNet learns all parameters simultaneously by minimizing a single loss function that weighs together different components of the network, and the choice of weights controls the tradeoff between learning either of the two nuisance functions. The related MADNet estimator by \citet{hines2025automaticdebiasingneuralnetworks} uses the same architecture as RieszNet, but it is trained by solving a moment-constrained minimization problem instead. Its shared layer is, however, still intended to be predictive of the Riesz representer. 

The logic behind outcome-adapted AutoDML estimation is rooted in the causal inference literature on efficient adjustment for estimation of the average treatment effect (ATE). Our outcome-adapted AutoDML estimator is a generalization of the debiased outcome-adapted propensity estimator (DOPE) proposed by \citet{christgau2025efficientadjustmentcomplexcovariates}, which is based on the classical augmented inverse propensity weighted (AIPW) estimator of the ATE \citep{robins_1995_aipw}. The ATE fits into the framework of AutoDML with the Riesz representer given in terms of the inverse propensity score, that is, the inverse probability of receiving treatment given the covariates. Expanding on work by \citet{Henckel_2022} and \citet{rotnitzky2019efficientadjustmentsetspopulation}, \citet{christgau2025efficientadjustmentcomplexcovariates} established a precursor of our main results specific to the AIPW estimator. They showed how to achieve an efficiency gain in terms of the asymptotic variance of the AIPW estimator by constructing the regression function and the propensity score as functions of a covariate representation that is predictive of the outcome but not the treatment. Our Theorem \ref{thm:delete_overadjust} generalizes Theorem 3.3(iii) by \citet{christgau2025efficientadjustmentcomplexcovariates} on relative efficiency from the AIPW estimator to general AutoDML estimators while at the same time weakening the assumptions. Similarly, our Theorem \ref{thm:supp_with_precision} generalizes their Corollary 3.9  for Riesz representers of a particular form, and it shows that any representation retaining predictive power of the Riesz representer can lead to an efficiency loss. 

Outcome-adapted AutoDML estimation differs sharply from previous AutoDML estimators by learning a covariate representation as part of an outcome regression first, and only subsequently is the Riesz representer learned for the particular representation. This procedure is not only supported by our theory and numerical experiments, but it also leads to a simpler estimator with fewer tuning parameters.

\subsection{Organization of the Paper}

The basic setup and a brief introduction to automatic debiased machine learning is given in Section \ref{sec:setup}. Section \ref{sec:efficient-covars} develops the theory of efficient covariate representations for automatic debiased machine learning. Section \ref{sec:OA_Autodml} defines the outcome-adapted automatic debiased machine learning estimator and establishes asymptotic results. Section \ref{sec:neural_nets} describes the neural network implementation of the outcome-adapted AutoDML estimator and compares it to RieszNet. Section \ref{sec:sim} reports on the setup of and results from our numerical experiments. All proofs and some auxiliary results and examples are in Appendices \ref{app:proofs} and \ref{sec: semiparametric efficiency}, and Appendix \ref{sec:sim_study_details} contains additional details regarding the numerical experiments.

\section{Setup} \label{sec:setup}
Let $\mathcal{P}$ be a statistical model consisting of distributions of a pair of random variables $(X,Y)\in\RR^d\times\RR$. We interpret $X$ as covariates, $Y$ as an outcome and assume that we observe $n\in\NN$ i.i.d. copies of $(X,Y)$.
Formally, we define $(X,Y)$ to be a mapping on a background measurable space $(\Omega,\mathbb{F})$, and we define a family of probability measures $(\mathbb{P}_P)_{P\in\mathcal{P}}$ on the background space such that the distribution of $(X,Y)$ is $P$ when the background space is equipped with $\mathbb{P}_P$.
We write $\mathbb{E}_P$ when taking expectations under $\mathbb{P}_P$ and assume throughout 
that $\EP{Y^2} < \infty$ for all $P \in \mathcal{P}$. When the distribution of $(X,Y)$ is $P$, we use $P_X$ and $P_Y$ to denote the marginal distributions of $X$ and $Y$, respectively. 

We let $\mathcal{M}^d$ be the space of real-valued measurable mappings on $\RR^d$ and, for $P\in\mathcal{P}$, 
$$
\mathcal{L}^2(P_X) = \left\lbrace\nu \in \mathcal{M}^d:\EP{\nu(X)^2}<\infty\right\rbrace.
$$
The corresponding Hilbert space consisting of $\mathcal{L}^2(P_X)$-equivalence classes is denoted $L^2(P_X)$. 

The parameter of interest is expressed as a functional on the statistical model, 
$$
    \psi:\mathcal{P}\to\RR.
$$
The parameters considered in this paper and estimated using AutoDML are those that can be expressed as a mean-square continuous linear functional evaluated at the regression function of $Y$ on $X$. We use $\gamma_P$ to denote the, almost surely unique, regression function under $P\in\mathcal{P}$, that is,
$$
    \gamma_P(x) = \EP{Y\mid X=x}. 
$$
\begin{definition}[AutoDML operators] \label{dfn-autoDML}
    We say that $M:\mathcal{M}^d \to \mathcal{M}^d$ is an AutoDML operator if, for all $P\in\mathcal{P}$,
    \begin{enumerate}[label=(\alph*)]
        \item \label{cond:lin} for $\nu_1,\nu_2\in\mathcal{L}^2(P_X)$, $a,b\in\RR$, $M(a\nu_1+b\nu_2) = aM(\nu_1)+bM(\nu_2)$ $P_X$-almost surely,
        \item \label{cond:msq} there exists $\kappa_P \geq 0$ such that for $\nu\in\mathcal{L}^2(P_X)$, 
        $$
        \EP{M(\nu)(X)^2} \leq \kappa_P \EP{\nu(X)^2}.
        $$
    \end{enumerate}
\end{definition}

Given an AutoDML operator, we define $m:\RR^d\times\mathcal{M}^d\to\RR$ by
$$
    m(x,\nu) = M(\nu)(x).
$$ 
The condition \ref{cond:msq} is equivalent to requiring that, for each $P \in \mathcal{P}$, there exists $\kappa_P \geq 0$ such that, for $\nu\in \mathcal{L}^2(P_X)$, 
$$
        \EP{m(X,\nu)^2} \leq \kappa_P \EP{\nu(X)^2}. 
$$
We follow the convention of \citet[][Assumption 1]{chernozhukov_24_via_riesz_regression} and refer to this property as \emph{mean-square continuity}. It ensures that if 
$\nu_1 = \nu_2$ $P_X$-a.s. then $M(\nu_1) = M(\nu_2)$ $P_X$-a.s., and conditions \ref{cond:lin} and \ref{cond:msq} together allow us to regard
$$
    M : L^2(P_X) \to L^2(P_X)
$$
as a well-defined bounded linear operator on the Hilbert space $L^2(P_X)$. 
We will also regard $m$ as a mapping $m: \RR^d \times L^2(P_X) \to \RR$, 
where for each $\nu \in L^2(P_X)$, $m(\cdot, \nu)$ is $P_X$-a.s. uniquely defined. 

\begin{definition}[AutoDML functionals]
    Let $M$ be an AutoDML operator. We say that $\psi:\mathcal{P}\to\RR$ given by
    $$
        \psi(P) = \EP{m(X,\gamma_P)}
    $$
    is an AutoDML functional. 
\end{definition}

Boundedness and linearity are the key properties of AutoDML operators as they allow us to obtain an $L^2(P_X)$ inner product representation of AutoDML functionals. Indeed, since mean-square continuity implies that the linear functional $\nu \mapsto \EP{m(X,\nu)}$ is continuous on $L^2(P_X)$, the following proposition is a direct consequence of the Riesz representation theorem.

\begin{proposition}\label{prop:existence-of-riesz-rep}
    Let $\psi$ be an AutoDML functional. For each $P\in\mathcal{P}$, there exists a unique $\alpha_P \in L^2(P_X)$, such that, for $\nu\in L^2(P_X)$,
    $$
        \EP{m(X,\nu)} = \EP{\alpha_P(X)\nu(X)}.
    $$
    In particular,
    $$
        \psi(P) = \EP{m(X,\gamma_P)} = \EP{\alpha_P(X)\gamma_P(X)}.
    $$
    We say that $\alpha_P$ is the Riesz representer associated with $\psi$. 
\end{proposition}

Below we provide examples of AutoDML functionals and their Riesz representers. For the first two functionals, we have the decomposition $X = (U,W)\in\RR\times\RR^{d-1}$ where we interpret $U$ as a treatment variable and $W$ as covariates. In the fourth example, we also have the decomposition $X = (U,W)\in\RR\times\RR^{d-1}$ where $U$ is a binary censoring indicator. 
\begin{example}[Average treatment effect]\label{ex-ate}
    With $U\in\{0,1\}$ being a binary treatment indicator,
    the average treatment effect (ATE) functional is given by
    $$
    \psi(P) = \EP{\gamma_P(1,W) - \gamma_P(0,W)}.
    $$
    Here, with $x=(u,w)$,
    $$
    m(x,\nu) = \nu(1,w)-\nu(0,w).
    $$
    With
    $$
    \pi_P(w) = \PP{U=1\mid W=w},
    $$
    the Riesz representer is
    $$
    \alpha_P(U,W)=\frac{U}{\pi_P(W)} - \frac{1-U}{1-\pi_P(W)}.
    $$
\end{example}

\begin{example}[Average shift effect]\label{ex-ase}
    With $U$ being a continuous treatment variable with conditional density w.r.t. Lebesgue measure $f_P(\cdot\mid w )$ given $W = w$, the average shift effect (ASE) functional is, for a fixed $\delta\in\RR$, given by
    $$
    \psi(P) = \EP{\gamma_P(U+\delta,W) - \gamma_P(U,W)}.
    $$
    Here, with $x=(u,w)$,
    $$
    m(x,\nu) = \nu(u + \delta,w)-\nu(u,w).
    $$
    The Riesz representer is
    $$
    \alpha_P(U,W)=\frac{f_P(U-\delta \mid W) - f_P(U\mid W)}{f_P(U\mid W)}.
    $$
\end{example}

\begin{example}[Average policy effect]\label{ex-ape}
    Let $X$ have density $f_P$ with respect to a measure $\mu$ on $\RR^d$ and let $f_1$ and $f_0$ be densities with respect to the same measure. The average policy effect is given by
    $$
    \psi(P) = \int \gamma_P(x) (f_1(x)-f_0(x)) \dd \mu(x).
    $$
    Here, whenever $\int |\nu(\tilde{x}) (f_1(\tilde{x})-f_0(\tilde{x}))| \dd \mu(\tilde{x}) < \infty$, we define
    \begin{align*}
    m(x,\nu) = 
     \int \nu(\tilde{x}) (f_1(\tilde{x})-f_0(\tilde{x})) \dd \mu(\tilde{x})
    \end{align*}
    and let $m(x, \nu) = 0$ otherwise.\footnote{If $\nu\in L^2(P_X)$ for some $P\in\mathcal{P}$, the regularity conditions which ensure mean-square continuity, see Lemma \ref{lem:int-treat-func}, also ensure $\int|\nu(\tilde x)(f_1(\tilde x)-f_0(\tilde x))|\dd\mu(\tilde x)<\infty$.} The Riesz representer is
    $$
    \alpha_P(X)=\frac{f_1(X)-f_0(X)}{f_P(X)}.
    $$
\end{example}

\begin{example}[Mean Missing Outcome]\label{ex:mmo}
    Let $\tilde Y$ be an uncensored outcome, $U \in \{0, 1\}$ a binary censoring indicator and 
    \[
        Y = \begin{cases}
            \tilde Y & \text{if} \quad U = 1,\\
            \texttt{NA} & \text{if} \quad U = 0.
        \end{cases}
    \]
    The mean missing outcome functional is given by
    $$
    \psi(P) = \EP{\gamma_P(1,W)}.
    $$
    Here, with $x=(u,w)$, 
    $$
    m(x,\nu) = \nu(1,w).
    $$
    With 
    $$
    \pi_P(w) = \PP{U=1\mid W=w},
    $$
    the Riesz representer is 
    $$
    \alpha_P(U,W) = \frac{U}{\pi_P(W)}. 
    $$
\end{example}

\subsection{The AutoDML Estimator}\label{seq-autodml}
The AutoDML estimator is based on the estimating function
$$
\varphi(X,Y,\gamma,\alpha,\psi) = m(X,\gamma) + \alpha(X)(Y-\gamma(X)) - \psi. 
$$
The crucial property of $\varphi$ is that it is \emph{doubly robust} in the sense that 
$$
\EP{\varphi(X,Y,\gamma,\alpha,\psi(P))} = 0
$$
if either $\gamma=\gamma_P$ or $\alpha=\alpha_P$. In addition, $\varphi(X,Y,\gamma_P,\alpha_P,\psi(P))$ is also the nonparametric efficient influence function for $\psi(P)$ as noted by \citet{newey1994}. For a precise statement and proof of this, see Appendix \ref{sec: semiparametric efficiency}. 

AutoDML estimators are constructed by first obtaining estimates $\hat\gamma$ and $\hat\alpha$ of $\gamma_P$ and $\alpha_P$, and then $\hat\psi$ is computed by solving the estimating equation determined by the estimating function. In this case,
$$
\hat\psi = \frac{1}{n}\sum_{i=1}^n m(X_i,\hat\gamma) + \hat\alpha(X_i)(Y_i-\hat\gamma(X_i)).
$$
Typically, cross-fitting (see, e.g., \citealt{chernozhukov2018doubledebiasedmachinelearningtreatment}) is used to ensure that the estimated functions $\hat\gamma$ and $\hat\alpha$ in the estimate of $\psi$ are evaluated only at observations not used in their construction. This leads to the following estimation procedure: 
\begin{enumerate}
    \item Let $I = \{1, \dots, n\}$ and let $D = ((X_i,Y_i))_{i \in I}$ be the full sample. Partition $I$ into $K$ disjoint sets $I_1,\dots, I_K$ of roughly equal size and set $D_k = ((X_i,Y_i))_{i \in I_k}$. 
    \item For $k=1, \dots ,K$ estimate $\hat\gamma_k,\hat\alpha_k$ using the observations not in $D_k$.
    \item Compute point and variance estimates
    \begin{align*}
        \hat\psi &= \frac{1}{n}\sum_{k=1}^K\sum_{i\in I_k} m(X_i,\hat\gamma_k) + \hat\alpha_k(X_i)(Y_i - \hat\gamma_k(X_i)),\\
        \hat V &= \frac{1}{n}\sum_{k=1}^K\sum_{i\in I_k} (m(X_i,\hat\gamma_k) + \hat\alpha_k(X_i)(Y_i - \hat\gamma_k(X_i)) - \hat\psi)^2. 
    \end{align*}
\end{enumerate}

Under regularity conditions on the model $\mathcal{P}$ and rate conditions on the estimators of $\gamma_P$ and $\alpha_P$, \citet{chernozhukov_24_via_riesz_regression} establish that the AutoDML estimator is a $\sqrt{n}$-consistent estimator with asymptotic variance
$$
V_P = \EP{\varphi(X,Y,\gamma_P,\alpha_P, \psi(P))^2} = \VP{m(X, \gamma_P)} + \EP{\alpha_P(X)^2\VP{Y \mid X}},    
$$
and that $\hat V$ is a consistent estimator of the asymptotic variance. 

To construct an AutoDML estimator, we require estimates of $\gamma_P$ and $\alpha_P$. As $\gamma_P$ is a regression function, it can be estimated using a variety of machine learning methods. To estimate $\alpha_P$, one option is to estimate the components of a closed-form expression for $\alpha_P$ and then plug these into the closed-form. For instance, for the ATE functional of Example~\ref{ex-ate}, we would estimate the propensity score $\pi_P$, which leads to the AIPW estimator of \citet{robins_1995_aipw}. 

AutoDML takes a different approach. The crucial observation at the foundation of AutoDML is that 
\begin{align*}
\alpha_P &= \argmin_{\alpha\in L^2(P_X)} \EP{(\alpha(X)-\alpha_P(X))^2} = \argmin_{\alpha\in L^2(P_X)} \EP{\alpha(X)^2- 2 \alpha(X) \alpha_P(X)}\\
&= \argmin_{\alpha\in L^2(P_X)} \EP{\alpha(X)^2 - 2 m(X,\alpha)},
\end{align*}
which suggests that $\alpha_P$ can be estimated directly by minimizing the Riesz loss
$$
\ell_{\mathrm{Riesz}}(\alpha) = \frac{1}{n}\sum_{i=1}^n \alpha(X_i)^2 - 2m(X_i,\alpha). 
$$
Minimizing the above loss function is referred to as Riesz regression, and a number of machine learning methods have been adapted to this task. These include LASSO \citep{Chernozhukov_2022_auto_dml_lasso}, random forests and neural nets \citep{chernozhukov2022riesznetforestrieszautomaticdebiased}, kernel ridge \citep{singh2024kernelridgerieszrepresenters}, and gradient boosting \citep{lee2025rieszboostgradientboostingriesz}. An advantage highlighted by \citet{chernozhukov_24_via_riesz_regression} of using Riesz regression, instead of estimating the components of a closed-form expression, is that Riesz regression avoids division by density or propensity score estimates, which can be unstable in the presence of small estimates.

\section{Efficient Covariate Representations}\label{sec:efficient-covars}
The central topic of this paper is efficient covariate representations. We define a covariate representation to be the result of applying a mapping to the covariates. 

\begin{definition}[Covariate representations] \label{def:covar_rep}
    Let $M: \mathcal{M}^d \to\mathcal{M}^d$ be an AutoDML operator and let \mbox{$h:\RR^d\to\RR^q$} be a measurable mapping. We say that $Z=h(X)$ is a permissible representation if there exists $M_h:\mathcal{M}^q\to\mathcal{M}^q$ such that, for all $\eta\in\mathcal{M}^q$,
    $$
    M_h(\eta)\circ h = M(\eta\circ h).
    $$
\end{definition}

Given a permissible representation, we define $m_h:\RR^q\times\mathcal{M}^q\to \RR$ by
$$
m_h(z,\eta) = M_h(\eta)(z), 
$$
and the condition in Definition \ref{def:covar_rep} is equivalent to requiring, for all $\eta\in\mathcal{M}^q$ and $x\in\RR^d$, with $z=h(x)$, that
$$
m_h(z,\eta) = m(x,\eta\circ h).
$$
For the functionals in Examples \ref{ex-ate}, \ref{ex-ase} and \ref{ex:mmo}, this condition is satisfied, for instance, if $h$ is of the form $h(u,w) = (u,\tilde h(w))$ such that the representation does not lose track of the treatment/censoring variable. For Example \ref{ex-ape}, any measurable $h$ is permissible.\footnote{See Example \ref{ex-not permissible} in Section \ref{seq: proofs sec:efficient-covars} for a representation that is not permissible.}

When $Z=h(X)$ is a permissible representation, we define 
\begin{align*}
    \gamma_{h,P}(z) & = \EP{Y\mid Z=z} \\
    \alpha_{h,P}(z) & = \EP{\alpha_P(X)\mid Z=z},
\end{align*}
and we can define an AutoDML functional in terms of $(Z,Y)$ and $M_h$. 
\begin{proposition}\label{prop:rep-functional}
    Let $\psi$ be an AutoDML functional determined by the AutoDML operator $M$, and let $h:\RR^d\to\RR^q$ induce a permissible representation $Z=h(X)$. Then $M_h$ is an AutoDML operator and
    $$
        \psi_h(P) = \EP{m_h(Z,\gamma_{h,P})}
    $$
    is an AutoDML functional with associated Riesz representer $\alpha_{h,P}$.
    Additionally, if either 
    $$
        \gamma_{h,P}(Z) = \gamma_{P}(X) \quad \text{ or } \quad \alpha_{h,P}(Z) = \alpha_{P}(X),
    $$
    then 
    $$
        \psi_h(P) = \psi(P).
    $$
\end{proposition}

Given a permissible representation, we define for $\gamma, \alpha \in \mathcal{M}^q$ the estimating function
$$
\varphi_h(Z,Y,\gamma,\alpha,\psi) = m_h(Z,\gamma) + \alpha(Z)(Y-\gamma(Z)) - \psi.
$$
Using $\varphi_h$ gives an AutoDML estimator of $\psi_h(P)$. Proposition \ref{prop:rep-functional} shows that if either the Riesz representer or the regression function is preserved by $h$, this AutoDML estimator is also an estimator of $\psi(P)$. As discussed in Section \ref{seq-autodml}, the estimator will, under regularity assumptions, have asymptotic variance equal to
$$
V_{h,P} = \EP{\varphi_h(Z,Y,\gamma_{h,P},\alpha_{h,P},\psi_h(P))^2}.
$$
Whether an AutoDML estimator based on $\varphi_h$ is preferable to one based on $\varphi$ depends on the relationship between $V_P$ and $V_{h,P}$. This leads us to the main theoretical result of the paper, showing that an AutoDML estimator based on a representation $Z$ that preserves predictive power of the outcome has smaller asymptotic variance than the estimator based on $X$.

\begin{theorem}\label{thm:delete_overadjust}
    Let $\psi$ be an AutoDML functional and let $h:\RR^d\to\RR^q$ induce a permissible representation $Z$. If 
    $$
    \gamma_{h,P}(Z) = \gamma_{P}(X) \quad \text{and} \quad \VP{Y \mid Z} = \VP{Y \mid X},
    $$
    then
    $$
    V_P - V_{h,P} = \EP{\VP{\alpha_P(X)\mid Z}\VP{Y\mid Z}} \geq 0. 
    $$
\end{theorem}

We will refer to representations that are adapted to the task of predicting outcomes as outcome-adapted. Theorem \ref{thm:delete_overadjust} not only shows us that outcome-adapted representations lead to smaller asymptotic variances, it also shows that to maximize the efficiency gain, we need the representation to carry as little information about $\alpha_P(X)$ as possible such that $\VP{\alpha_P(X)\mid Z}$ is large.

Note that Example 3.10 by \citet{christgau2025efficientadjustmentcomplexcovariates} shows that we cannot drop the assumption $\VP{Y \mid Z} = \VP{Y \mid X}$ in Theorem \ref{thm:delete_overadjust}, and our proof reveals how this is the natural minimal requirement of a representation, in addition to $\gamma_{h,P} = \gamma_P$, to ensure $V_{h,P} \leq V_P$.

\subsection{Riesz-adapted Covariate Representations}\label{sec:supp_precision}

Proposition \ref{prop:rep-functional} shows that $\psi(P)=\psi_h(P)$ if $\alpha_{h,P}(Z) = \alpha_{P}(X)$. Thus, it is natural to ask what happens if we use a Riesz-adapted representation in the spirit of DragonNet \citep{shi2019adapting} and RieszNet \citep{chernozhukov2022riesznetforestrieszautomaticdebiased}. 
It turns out that for a particular\footnote{Example \ref{ex: counterexample ase.} in Section \ref{seq: proofs sec:efficient-covars} shows that there exist cases where $V_P >  V_{h,P}$ while $\alpha_{h,P}(Z) = \alpha_{P}(X)$.} but broad class of functionals, $V_P \leq V_{h,P}$ if the 
representation preserves the Riesz representer, that is, if $\alpha_{h,P}(Z) = \alpha_{P}(X).$ 

\begin{lemma}\label{lem:int-treat-func}
    Assume that $X=(U,W)\in\RR^{d-k}\times\RR^{k}$ and that $U$ has conditional density $f_P(\cdot \mid w)$ given $W = w$ with respect to a $\sigma$-finite measure $\mu$ on $\RR^{d-k}$. Let $g:\RR^{d-k}\to\RR$ be measurable and define, whenever $\int |g(\tilde u) \nu(\tilde u,w)| \dd\mu(\tilde u)<\infty$, $m:\RR^d \times \mathcal{M}^d\to\RR$ by 
    $$
    m((u,w),\nu) = 
    \int g(\tilde u) \nu(\tilde u,w) \dd\mu(\tilde u)
    $$
    and $m((u, w), \nu) = 0$ otherwise. If, for all $P\in\mathcal{P}$, there exists $\kappa_P > 0$ such that, for $P_W \otimes \mu$-almost all $(u,w)\in\RR^{d-k}\times\RR^{k}$,
    \begin{equation}\label{eq:bound}
    \left|\frac{g(u)}{f_P(u\mid w)}\right|\leq \kappa_P
    \end{equation}
    then $\psi(P) = \EP{m(X,\gamma_P)}$ is an AutoDML functional with Riesz representer 
    $$
        \alpha_P(U,W) = \frac{g(U)}{f_P(U\mid W)}.
    $$
    Moreover, if $h:\RR^d\to\RR^q$ is of the form $h(u,w) = (u,\tilde h(w))$, then $Z = (U,\tilde h(W)) = (U,\tilde Z)$ is a permissible representation.  
\end{lemma}

Examples of functionals that are of the type specified in Lemma \ref{lem:int-treat-func} include the ATE from Example \ref{ex-ate}, the average policy effect from Example \ref{ex-ape}, the mean missing outcome from Example \ref{ex:mmo}, and the weighted average derivative functional discussed by \citet{Chernozhukov_2022_auto_dml_lasso}. In the case of the ATE, the condition \eqref{eq:bound} is equivalent to the propensity scores being bounded away from 0 and 1. For the particular functionals and representations in Lemma \ref{lem:int-treat-func}, we have a converse of Theorem \ref{thm:delete_overadjust}, which shows that using a Riesz-adapted representation leads to a larger asymptotic variance. 

\begin{theorem}\label{thm:supp_with_precision}
    Let $\psi$ be an AutoDML functional and $Z = h(X)$ a permissible representation as in Lemma \ref{lem:int-treat-func}. If 
    $$
        \alpha_{h,P}(Z) = \alpha_{P}(X), 
    $$
    then
    $$
        V_{h,P} - V_P = \VP{R_P} \geq 0,
    $$
    where 
    $$
    R_P = \varphi_h(Z,Y,\gamma_{h,P},\alpha_{h,P},\psi(P)) - \varphi(X,Y,\gamma_P,\alpha_P,\psi(P)). 
    $$
\end{theorem}

\section{The Outcome-adapted AutoDML Estimator}\label{sec:OA_Autodml}
Theorem \ref{thm:delete_overadjust} shows that an AutoDML estimator based on an outcome-adapted representation $Z=h(X)$ can have smaller asymptotic variance than an AutoDML estimator based on $X$. The theorem additionally shows that to maximize the efficiency gain, the representation should discard as much information about $\alpha_P(X)$ as possible. In practice, any such representation is typically unknown. Thus, we have to estimate an outcome-adapted representation, $\hat{h}$, to turn the theoretical results into an estimation procedure. We can then subsequently construct an AutoDML estimator using $\varphi_{\hat h}$. We refer to this procedure as the outcome-adapted AutoDML estimator.

\begin{algorithm} 
\caption{The sample split outcome-adapted AutoDML estimator}
\label{alg:oa-autodml}
\begin{algorithmic}[1]
\Require An Auto-DML operator, a set, $\mathcal{H}$, of admissible representations, and a dataset $(X_1,Y_1), \ldots, (X_n, Y_n)$.
\State Partition $\{1, \dots n\}$ into two disjoint sets $I_1$ and $I_2$ of fixed proportion.
\State Estimate a regression function  of the form  $\hat{\gamma} = \hat{\eta} \circ \hat{h}$ with 
$\hat{h} \in \mathcal{H}$ via outcome regression using $((X_i,Y_i))_{i \in I_1}.$
\State Compute $\hat Z_i = \hat h(X_i)$ for $i = 1, \dots, n$.
\State Compute the estimate $\hat\gamma_{\hat{h}}$ by regression using $((\hat Z_i,Y_i))_{i\in I_1}$.
\State Compute the estimate $ \hat\alpha_{\hat{h}}$ by Riesz regression using $(\hat Z_i)_{i \in I_1}$.
\State Compute
\Statex
\begin{align*}
    \hat\psi
    &= \frac{1}{|I_2|}\sum_{i\in I_2}
    m_{\hat{h}}(\hat Z_i, \hat \gamma_{\hat{h}})
    + \hat\alpha_{\hat{h}}(\hat Z_i)
    \bigl(Y_i - \hat \gamma_{\hat{h}}(\hat Z_i)\bigr), \\
    \hat V
    &= \frac{1}{|I_2|}\sum_{i\in I_2}
    \left(
    m_{\hat{h}}(\hat Z_i, \hat \gamma_{\hat{h}})
    + \hat\alpha_{\hat{h}}(\hat Z_i)
    \bigl(Y_i - \hat \gamma_{\hat{h}}(\hat Z_i)\bigr)
    - \hat\psi
    \right)^2.
\end{align*}
\State \Return $\hat{\psi}, \hat{V}$
\end{algorithmic}
\end{algorithm}
We note that it is possible to take $\hat{\gamma}_{\hat{h}} = \hat{\eta}$ in Step 4 of the algorithm, where $\hat{\eta}$ is computed as part of Step 2. Thus, Step 4 is an optional fine tuning step of the outcome regression beyond estimation of the representation. In Section \ref{sec:neural_nets}, we propose a particular implementation of the algorithm using neural networks where $\hat{\gamma}_{\hat{h}} = \hat{\eta}$, and in our concrete implementation, $\hat{h}$ is composed of three hidden layers of the network. An alternative implementation could take $\mathcal{H}$ as a set of coordinate projections of the covariate space, and Step 2 could be a simple variable selection step, while Step 4 would then be the actual outcome regression using the selected variables, see Section \ref{seq:representation error}.

There is a loss of efficiency due to the simple sample splitting procedure relative to using cross-fitting, as the variance of $\hat{\psi}$ will scale with $|I_2|$ rather than $n$. In practice, the outcome-adapted AutoDML estimator could be implemented with cross-fitting or without sample splitting entirely. However, using cross-fitting would substantially complicate the asymptotic arguments in the sequel, as a different representation would be estimated for each cross-fitting fold.

The outcome-adapted AutoDML estimator is a generalization of the debiased outcome-adapted propensity estimator proposed by \citet{christgau2025efficientadjustmentcomplexcovariates} to the framework of automatic debiased machine learning.

\subsection{Asymptotics} \label{sec-asymp}
Since the representation $\hat{h}$ is estimated, we generally have $\psi(P) \neq \psi_{\hat{h}}(P).$
This motivates the error decomposition,
$$
\hat{\psi} - \psi(P) =  \underbrace{\hat{\psi}-\psi_{\hat{h}}(P)}_{\text{sampling error}} + \underbrace{\psi_{\hat{h}}(P)-\psi(P)}_{\text{rep. error}},
$$
and we refer to the first term as the sampling error and the second as the representation error. We analyze these separately below.

We state our asymptotic results as uniform distributional convergence over the statistical model $\mathcal{P}$,
see, e.g., Section 1.2 by \citet{Shah_2020} for an argument of why uniform distributional convergence is preferable to pointwise convergence. To this end, we use the following notation. Let $T_1,T_2,\dots$ be a sequence of random variables on $(\Omega,\mathbb{F})$, let $a_n > 0$ be a sequence of real numbers, and let $\Phi$ be the distribution function of the standard Gaussian distribution. If, for all $\varepsilon>0$, 
$$
\lim_{n\to\infty} \sup_{P\in\mathcal{P}}\PP{|T_n/a_n|>\varepsilon} = 0,
$$
we say that $T_n$ tends to $0$ in probability uniformly over $\mathcal{P}$ at rate (faster than) $a_n$ and write 
$$
T_n = o_\mathcal{P}(a_n).
$$
If
$$
\lim_{n\to\infty}\sup_{P\in\mathcal{P}}\sup_{t\in\RR} |\PP{T_n\leq t} - \Phi(t)| = 0,
$$
we say that $T_n$ converges in distribution to a standard Gaussian uniformly over $\mathcal{P}$ and write 
$$
T_n\stackrel{d/\mathcal{P}}{\to}\mathcal{N}(0,1).
$$

\subsection{Sampling Error}\label{seq:sampling error}

For our uniform asymptotic results, we first require a uniform version of the mean-square continuity that we assume in Definition~\ref{dfn-autoDML}.

\begin{assumption}[Uniform mean-square continuity]\label{assump-mean-square continuity}
    There exists $\kappa > 0$ such that for all $P \in \mathcal{P}$, it holds for all $\nu\in L^2(P_X)$,
    $$
    \EP{m(X,\nu)^2} \leq \kappa \EP{\nu(X)^2}.
    $$
\end{assumption}

We also need to impose some additional standard assumptions on our data-generating processes, our set of permissible representations and the outcome and Riesz regression estimators. 

\begin{assumption}\label{assump-asymptotics}
Let $D_1 = ((X_i, Y_i))_{i \in I_1}$ denote the data used for estimating the nuisance functions in Algorithm~\ref{alg:oa-autodml}. Let $X \sim P_X$ be independent of $D_1$ and $\hat{Z} = \hat{h}(X)$.
    \begin{enumerate}[label=(\roman*)]
        \item There exist $C, \delta, c > 0$ such that, for all $P \in \mathcal{P}$, we have $\EP{|Y|^{2+\delta}} \leq C$ and, for all $h \in \mathcal{H}$, we have, with $Z = h(X)$, $\EP{\varphi_{h,P}(Z,Y)^2} \geq c$ and
        \[
        \max \left\lbrace  |\alpha_{h, P}(Z)|,  \EP{\left(Y-\gamma_{h,P}(Z)\right)^{2} \mid Z}, \EP{\left|m_h(Z, \gamma_{h,P})-\psi_h(P)\right|^{2+\delta}}  \right\rbrace \leq C.
        \]
        In addition, $|\hat{\alpha}_{\hat{h}}(\hat{Z})| \leq C$.
        \item $\EP{\left(\hat\gamma_{\hat{h}}(\hat{Z}) - \gamma_{\hat{h}, P}(\hat{Z})\right)^2 \mid D_1} = o_\mathcal{P}(1)$,
        \item $\EP{\left(\hat\alpha_{\hat{h}}(\hat{Z}) - \alpha_{\hat{h}, P}(\hat{Z})\right)^2 \mid D_1} = o_\mathcal{P}(1)$,
        \item $n^{1/2}\EP{\left(\hat\alpha_{\hat{h}}(\hat{Z}) - \alpha_{\hat{h}, P}(\hat{Z})\right)^2 \mid D_1}^{1/2}     \EP{\left(\hat\gamma_{\hat{h}}(\hat{Z}) - \gamma_{\hat{h}, P}(\hat{Z})\right)^2 \mid D_1}^{1/2} = o_\mathcal{P}(1)$.
    \end{enumerate}
\end{assumption}

Assumption~\ref{assump-asymptotics}~(i) consists of mild regularity conditions and is mostly benign except for highly expressive representations (large $\mathcal{H}$). Assumption~\ref{assump-asymptotics}~(ii)-(iv) are conditions ensuring sufficiently fast convergence of the nuisance estimates and are reminiscent of other conditions imposed in the literature on doubly-robust estimators. Overall, these assumptions are similar to Assumptions 3 and 4 made by \citet{chernozhukov_24_via_riesz_regression}. 

We now state the main asymptotic result for the sampling error.

\begin{theorem}\label{thm-inferens}
    Under Assumptions \ref{assump-mean-square continuity} and \ref{assump-asymptotics}, the outcome-adapted AutoDML estimator satisfies
    $$
    \left(V_{\hat h,P}/|I_2|\right)^{-1/2}(\hat\psi - \psi_{\hat{h}}(P))\stackrel{d/\mathcal{P}}{\to}\mathcal{N}(0,1).
    $$
    In addition, it holds that 
    $$
    \hat{V} = V_{\hat{h},P} + o_{\mathcal{P}}(1).
    $$
\end{theorem}

\begin{corollary}\label{corr-inference}
    Under Assumptions \ref{assump-mean-square continuity} and \ref{assump-asymptotics},
    $$
    \left(\hat V/|I_2|\right)^{-1/2}(\hat\psi - \psi_{\hat{h}}(P))\stackrel{d/\mathcal{P}}{\to}\mathcal{N}(0,1).
    $$
    In particular, it holds for $t\in(0,1)$,
    $$
    \sup_{P\in\mathcal{P}} \left| P\left(\psi_{\hat{h}}(P) \in \hat\psi \, \pm\,  q_{1-t/2}\left(\hat V/|I_2|\right)^{1/2} \right) - (1-t) \right|\to 0,
    $$
    where $q_{1-t/2}$ is the $1-t/2$-quantile of a standard normal distribution. 
\end{corollary}

Consequently, conditionally on $\hat{h}$, the outcome-adapted AutoDML estimator approximately has a Gaussian distribution with mean $\psi_{\hat{h}}(P)$ and variance $V_{\hat{h},P} / |I_2|$.

\subsection{Representation Error and Variable Selection}\label{seq:representation error}

With regard to the representation error, we have the following result.

\begin{proposition}\label{thm-rep-error}
    Let $\kappa_P \geq 0$ be the constant from Definition \ref{dfn-autoDML}. If $X \sim P_X$ independent of $D_1$ and $\hat{Z} = \hat{h}(X)$, then 
    $$
    \left|\psi_{\hat{h}}(P) - \psi(P)\right| \leq \kappa_P^{1/2} \EP{\left(\gamma_{\hat{h},P}(\hat{Z}) - \gamma_P(X)\right)^2\mid D_1}^{1/2}.
    $$ 
\end{proposition}

Proposition \ref{thm-rep-error} shows that if we, in addition to Assumptions \ref{assump-mean-square continuity} and \ref{assump-asymptotics}, also have
\begin{equation}\label{eq-representation convergence rate}
\EP{\left(\gamma_{\hat{h},P}(\hat{Z}) - \gamma_P(X)\right)^2\mid D_1}^{1/2} = o_\mathcal{P}(n^{-1/2}), 
\end{equation}
then Corollary \ref{corr-inference} implies that the Gaussian confidence interval
$$
\hat{\psi} \pm 1.96 \cdot \sqrt{\hat{V}/|I_2|}
$$
will have 95\% asymptotic coverage for $\psi(P)$ uniformly over $\mathcal{P}$. 

To demonstrate that \eqref{eq-representation convergence rate} is achievable, we consider 
a variable selection setup. To this end, suppose $\gamma_P(x) = \eta_P (\beta_P^\top x)$ for $\beta_P \in \mathbb{R}^d$ and $\eta_P : \mathbb{R} \to \mathbb{R}$,
and that the permissible representations in $\mathcal{H}$ are all coordinate projections defined on $\mathbb{R}^d$. We define the support of $\gamma_P$ as
$$
S(P) = \{i \in \{1, \ldots, d\} : \beta_{P, i} \neq 0\},
$$
and the corresponding oracle representation $h_P : \mathbb{R}^d \to \mathbb{R}^{S(P)}$ as the 
projection onto the coordinates corresponding to the indices in $S(P)$. Then $\gamma_P = \gamma_{h_P, P} \circ h_P$ with 
$$
    \gamma_{h_P, P}((x_i)_{i \in S(P)}) = \eta_P \left(\sum_{i \in S(P)} \beta_{P, i} x_i\right).
$$
The estimated coordinate projection, $\hat{h} : \mathbb{R}^d \to \mathbb{R}^{\hat{S}}$, is determined entirely 
by the set $\hat{S} \subseteq  \{1, \ldots, d\}$ of selected variables. Within this setup, we show the following corollary of Proposition \ref{thm-rep-error}.

\begin{corollary} \label{cor-var-select}
With $\mathcal{H}$ a set of permissible coordinate projections on $\mathbb{R}^d$, then 
$$
    \sup_{P \in \mathcal{P}} \mathbb{P}_P(\psi_{\hat{h}}(P) \neq \psi(P)) \leq 
    1 - \inf_{P \in \mathcal{P}} \mathbb{P}_P(S(P) \subseteq \hat{S}).
$$
In particular, if
\begin{equation}
    \label{eq:uniform-SURE}    
        \inf_{P \in \mathcal{P}} \mathbb{P}_P(S(P) \subseteq \hat{S}) \to 1
\end{equation}   
then
$$
    \left|\psi_{\hat{h}}(P) - \psi(P)\right| = o_\mathcal{P}(a_n)
$$
for any sequence $a_n > 0$.
\end{corollary}

The property \eqref{eq:uniform-SURE} is a uniform version of the sure screening property introduced by \citet{fanSureIndependenceScreening2008}. When $\eta_P$ is the identity, their Sure Independence Screening (SIS) algorithm has the sure screening property \citep[Theorem 1]{fanSureIndependenceScreening2008} under their conditions 1--4, and if $P \in \mathcal{P}$ fulfills a uniform version of these conditions, 
\eqref{eq:uniform-SURE} holds for the SIS algorithm. The high dimensional ordinary least squares projection (HOLP) estimator by \citet{wangHighDimensionalOrdinary2016} is another example of a variable selection method with the sure screening property. When our representation learning step consists of variable selection, our outcome-adapted AutoDML estimator becomes analogous to the outcome-adaptive LASSO proposed by \citet{OAL}, which selects covariates for propensity score estimation as part of estimating ATE. 

 Without additional assumptions it seems difficult to relax the strong condition \eqref{eq-representation convergence rate} to ensure that the representation error vanishes, and, unless all representations have a particular structure, such as being coordinate projections, the required $\sqrt{n}$-rate might be unattainable. Note, however, that \eqref{eq-representation convergence rate} is a condition on the information that $\hat{h}(X)$ carries about $Y$ relative to $X$, and not a condition on the convergence rate of the procedure of regressing $Y$ onto $\hat{h}(X)$. Condition \eqref{eq-representation convergence rate} is, nevertheless, strong, and if we cannot justify that the representation error is negligible for a practical application,
 it is advisable to compute confidence intervals via bootstrapping for the general outcome-adapted AutoDML estimator as proposed by \citet{christgau2025efficientadjustmentcomplexcovariates} for DOPE.

\section{Outcome-adapted Neural Networks}\label{sec:neural_nets}
One way of constructing an outcome-adapted representation is to regress outcomes onto covariates using a feed-forward neural network and use the output of one of the hidden layers as the representation. We can then subsequently use this representation to estimate the Riesz representer.

To be specific, we consider a neural network, as depicted in Figure \ref{fig:riesz_net_architecture}, of arbitrary depth with arbitrary activation functions, consisting of a shared representation trunk followed by an outcome branch with a linear output layer and a Riesz representer branch, also with a linear output layer. We introduce the two loss functions
\begin{align*}
        \ell_\text{MSE}(\eta_\gamma, h) & = \frac{1}{n}\sum_{i=1}^n (Y_i- \eta_\gamma(h(X_i)))^2, \\
        \ell_{\text{Riesz}}(\eta_\alpha, h) & = \frac{1}{n} \sum_{i=1}^n \eta_\alpha(h(X_i))^2 - 2 m_h(h(X_i), \eta_\alpha),
\end{align*}
with the regression function $\gamma = \eta_\gamma \circ h$ and the Riesz representer $\alpha = \eta_\alpha \circ h$ sharing the representation trunk $h : \mathbb{R}^d \to \mathbb{R}^q$, and with \mbox{$\eta_\gamma, \eta_\alpha : \mathbb{R}^q \to \mathbb{R}$.} We denote by $w_h$, $w_\gamma$, and $w_\alpha$ the parameters associated with the representation trunk, the outcome branch, and the Riesz representer branch, respectively. We train the network by a two-step procedure:

\begin{enumerate}
    \item Optimize $w_h$ and $w_\gamma$ by minimizing the mean-square loss $\ell_\text{MSE}(\eta_\gamma, h)$.
    \item Freeze $w_h$ and optimize $w_\alpha$ by minimizing the Riesz loss $\ell_{\text{Riesz}}(\eta_\alpha, \hat{h})$.
\end{enumerate}

As we fix the representation after estimating the regression function, the intended effect is that the representation retains information that is relevant for predicting outcomes. We refer to a neural network trained in this manner as an outcome-adapted neural network. We could also obtain a Riesz-adapted neural network by reversing the above training procedure and learning $w_h$ as part of the Riesz regression. 

To maximize the variance reduction in Theorem \ref{thm:delete_overadjust}, we need the representation to carry as little information as possible about $\alpha_P(X)$. To achieve this, we consider two implementations of an information bottleneck at the shared representation layer.

\begin{itemize}
    \item \textbf{Simple dimensionality reduction.} We can enforce an information bottleneck by choosing a small dimension for the shared representation $Z$. In practice, we would choose the dimension $q$ of $Z$ by cross-validation during the first step of the training procedure. 
    \item \textbf{Adaptive dimensionality reduction.} Instead of forcing $Z$ to have a particular dimension, we can encourage the neural network to learn a low-dimensional representation $\hat{Z}$. We do this using a group LASSO type penalty on the weights in the $Z$-layer of the neural network. Specifically, with $\tilde w_j$ being the weights on the edges pointing to the $j$-th neuron of the layer (the $j$-th coordinate of $Z$), we add the following penalty to the MSE loss when optimizing it in the first step of the training procedure,
    $$
    \lambda \sum_{j} \| \tilde w_j \|_2.
    $$
    Here, $\lambda>0$ is a tuning parameter controlling the magnitude of the penalty. This type of group LASSO penalty has the ability to set entire neurons to zero leading to a dimensionality reduction. The parameter $\lambda$ can be chosen by cross-validation.
\end{itemize}

\subsection{Relationship to RieszNet}\label{seq: RieszNet}
Our architecture is inspired by the RieszNet architecture developed by \citet{chernozhukov2022riesznetforestrieszautomaticdebiased}. RieszNet is also a neural network that uses a shared representation for estimation of the outcome and Riesz representer. The main conceptual difference is that the representation in RieszNet is intended to be Riesz-adapted. As part of our simulation study, we will examine the relationship between the outcome-adapted neural network we propose and RieszNet, and we therefore include a description of their implementation differences.

There are two implementation differences between an outcome-adapted neural network and RieszNet. Firstly, RieszNet introduces an additional parameter $\epsilon \in \RR$ that modifies the regression function as
$$
\tilde{\gamma}(X) = \gamma(X) + \epsilon \alpha(X).
$$
During training, RieszNet incorporates an additional so-called TMLE loss motivated by the TMLE framework proposed by \citet{vanderLaanRose2011}. The TMLE loss is a mean-square error loss on the modified regression function, 
$$
\ell_{\text{TMLE}}(\gamma, \alpha, \epsilon) = \frac{1}{n} \sum_{i=1}^n (Y_i- \tilde{\gamma}(X_i))^2.
$$
Secondly, the training procedure is different. RieszNet optimizes all parameters, $w_h$, $w_\gamma$, $w_\alpha$, and $\epsilon$, simultaneously by minimizing the loss
$$
\ell_{\text{RieszNet}} = \lambda_\text{MSE}\ell_{\text{MSE}} + \lambda_\text{Riesz} \ell_{\text{Riesz}} + \lambda_\text{TMLE} \ell_{\text{TMLE}},
$$
where $\lambda_\text{MSE}, \lambda_\text{Riesz}, \lambda_\text{TMLE} \geq 0$ are user-chosen parameters. \citet{chernozhukov2022riesznetforestrieszautomaticdebiased} take $\lambda_\text{MSE} = 1$, and $\lambda_\text{Riesz}$ and $\lambda_\text{TMLE}$ then weigh the individual loss components relative to the MSE loss. Figure \ref{fig:riesznet_comparison}  illustrates the similarities and differences between the outcome-adapted neural network and RieszNet.
\begin{figure}
\centering
\begin{tikzpicture}[
    font=\small,
    >=stealth,
    sharedcolor/.style={fill=blue!25},
    outcomecolor/.style={fill=red!25},
    rieszcolor/.style={fill=orange!25},
    tmlecolor/.style={fill=yellow!30},
    layer/.style={
        rectangle,
        draw,
        rounded corners,
        minimum width=2.2cm,
        minimum height=1.0cm,
        align=center
    },
    io/.style={
        circle,
        draw,
        minimum size=0.9cm,
        align=center
    },
    procbox/.style={
        rectangle,
        draw,
        rounded corners,
        fill=gray!8,
        text width=6.4cm,
        minimum height=3.8cm,
        align=left,
        inner sep=4pt
    },
    lossbox/.style={
        rectangle,
        draw,
        rounded corners,
        minimum width=2.8cm,
        minimum height=1.2cm,
        align=center,
        inner sep=4pt,
        font=\small
    }
]


\node[io, fill=green!25] (X) at (0,0) {$X$};

\node[layer, sharedcolor] (T1) at (2.2,0)
{Hidden\\Layers};

\node[io, sharedcolor] (Z) at (4.4,0)
{$Z$};

\coordinate (split) at (5.7,0);

\node[layer, outcomecolor] (G1) at (7.6,1.5)
{Hidden\\Layers};

\node[io, outcomecolor] (gamma) at (9.9,1.5)
{$\gamma$};

\node[layer, rieszcolor] (A1) at (7.6,-1.5)
{Hidden\\Layers};

\node[io, rieszcolor] (alpha) at (9.9,-1.5)
{$\alpha$};



\node[
    rectangle,
    draw=black,
    rounded corners,
    fill=yellow!30,
    inner sep=5pt
]
(eqbox)
at (12.4,0)
{
$\tilde{\gamma} = \gamma + \epsilon \alpha$
};


\node[
    font=\bfseries,
    blue!70!black
]
at (3.3,1.0)
{Shared Trunk};

\node[
    font=\bfseries,
    red!70!black
]
at (8.7,2.25)
{Outcome Branch};

\node[
    font=\bfseries,
    orange!80!black
]
at (8.7,-2.25)
{Riesz Branch};


\draw[purple!70!black, thick]
(-0.3,3.2)
--
(10.2,3.2);

\draw[purple!70!black, thick]
(-0.3,3.2)
--
(-0.3,2.8);

\draw[purple!70!black, thick]
(10.2,3.2)
--
(10.2,2.8);

\node[
    purple!70!black,
    font=\bfseries
]
at (4.9,3.55)
{Shared Across Both Architectures};

\draw[yellow!50!black, thick]
(10.6,3.2)
--
(13.3,3.2);

\draw[yellow!50!black, thick]
(10.6,3.2)
--
(10.6,2.8);

\draw[yellow!50!black, thick]
(13.3,3.2)
--
(13.3,2.8);

\node[
    yellow!50!black,
    font=\bfseries,
    align=center
]
at (11.95,3.7)
{RieszNet TMLE \\Correction};


\draw[->, thick] (X) -- (T1);
\draw[->, thick] (T1) -- (Z);

\draw[thick] (Z.east) -- (split);

\filldraw[black] (split) circle (2pt);

\draw[->, thick]
(split)
|- (G1.west);

\draw[->, thick]
(split)
|- (A1.west);

\draw[->, thick] (G1) -- (gamma);
\draw[->, thick] (A1) -- (alpha);

\draw[->, thick]
(gamma.east)
-- ++(0.6,0)
|- (eqbox.north west);

\draw[->, thick]
(alpha.east)
-- ++(0.6,0)
|- (eqbox.south west);


\node[
    draw,
    dashed,
    rounded corners,
    inner sep=0.65cm,
    fit=(X)(T1)(Z)(G1)(gamma)(A1)(alpha)(eqbox)
]
(archbox)
{};


\node[procbox]
(procA)
at (3,-5.5)
{
\centering
{\large \textbf{Outcome-adapted Training}}

\begin{enumerate}
    \item Train {\color{blue!70!black}shared trunk}
    and {\color{red!70!black}outcome branch}
    with $\ell_{\text{MSE}}$.

    \item Freeze {\color{blue!70!black}shared trunk}
    and train {\color{orange!80!black}Riesz branch}
    with $\ell_{\text{Riesz}}$.
\end{enumerate}
};

\node[procbox]
(procB)
at (10.7,-5.5)
{
\centering
{\large \textbf{RieszNet Training}}

\begin{enumerate}
    \item Choose weights
    $\lambda_{\text{MSE}},
    \lambda_{\text{Riesz}},
    \lambda_{\text{TMLE}} \geq 0$.

    \item Train
    {\color{blue!70!black}shared trunk}, {\color{red!70!black}outcome branch}, {\color{orange!80!black}Riesz branch} and {\color{yellow!50!black}$\epsilon$} with
    \begin{align*}
    \ell_{\text{RieszNet}}
    &=
    \lambda_{\text{MSE}}\ell_{\text{MSE}}
    \\ &\,+
    \lambda_{\text{Riesz}}\ell_{\text{Riesz}}
    \\ &\,+
    \lambda_{\text{TMLE}}\ell_{\text{TMLE}}
    \end{align*}
\end{enumerate}
};

\node[align = left,
](losses)
at (2.2,-1.6)
{\textbf{Loss Components} \\
\color{red!70!black}{$\ell_{\text{MSE}} = \sum_{i=1}^n (Y_i -\gamma(X_i))^2$} \\
\color{orange!80!black}{$\ell_{\text{Riesz}} = \sum_{i=1}^n \alpha(X_i)^2 -2m(X_i,\alpha)$} \\ 
\color{yellow!50!black}{$\ell_{\text{TMLE}} = \sum_{i=1}^n (Y_i -\tilde{\gamma}(X_i))^2$}};


\draw[->, thick]
($(archbox.south)+(-3.5,0)$)
--
(procA.north);

\draw[->, thick]
($(archbox.south)+(4.2,0)$)
--
(procB.north);

\end{tikzpicture}
\caption{
Comparison of architectures and training procedures for the outcome-adapted neural network and RieszNet.
}
\label{fig:riesznet_comparison} 

\end{figure}
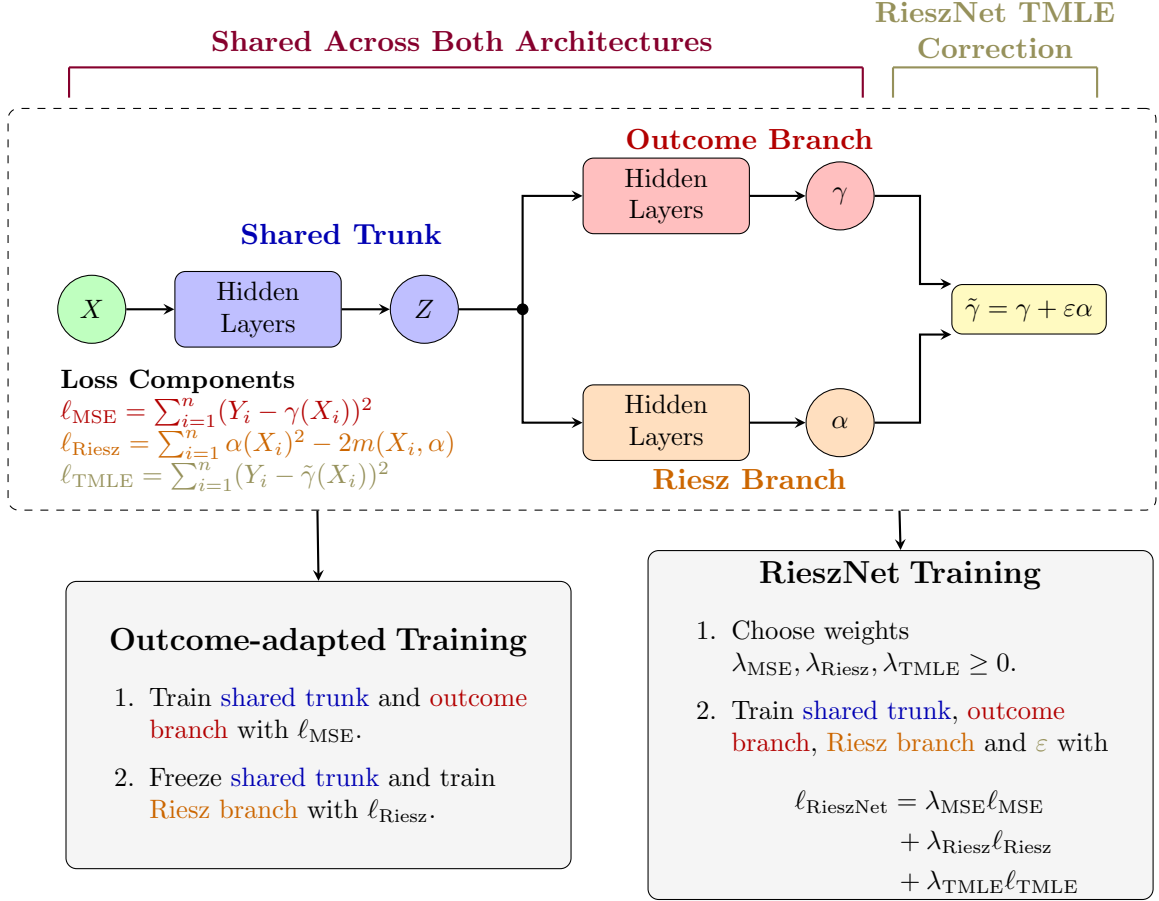

\citet{chernozhukov2022riesznetforestrieszautomaticdebiased} argue that they ``can give special attention to features that are predictive of the Riesz representer'' when estimating the outcome regression function. This argument is based on their Lemma 3.1, which is a special case of our Proposition \ref{prop:rep-functional} with $Z = \alpha_P(X)$. It justifies that the resulting AutoDML estimator is an estimator of $\psi(P)$, but in light of Theorem \ref{thm:supp_with_precision}, an estimator based on the representation $Z = \alpha_P(X)$ will be among the most inefficient AutoDML estimators. If the shared representation really was intended to be Riesz-adapted, RieszNet performs remarkably well in the experiments conducted by \citet{chernozhukov2022riesznetforestrieszautomaticdebiased}. We have two main hypotheses regarding the actual performance of RieszNet.
\begin{enumerate}
    \item The choice of $\lambda_\text{Riesz}$ determines whether the representation used by RieszNet is outcome- or Riesz-adapted. Accordingly, for $\lambda_\text{Riesz} \to 0$, RieszNet will behave similarly to the outcome-adapted neural network, and for $\lambda_\text{Riesz} \to \infty$ RieszNet will behave similarly to a Riesz-adapted neural network. The reason RieszNet is found to perform well by \citet{chernozhukov2022riesznetforestrieszautomaticdebiased} is that $\lambda_\text{Riesz}$ is chosen fairly small.
    \item The TMLE loss is essentially a mean-square error loss on outcome predictions. Thus, the TMLE loss contributes to making the shared representation outcome-adapted and RieszNet allows for a tradeoff between the weight on the MSE loss and the TMLE loss. That is, keeping $\lambda_\text{Riesz}$ fixed, using the weights $(\lambda_\text{MSE},\lambda_\text{Riesz},\lambda_\text{TMLE})$ and $(\tilde{\lambda}_\text{MSE}, \lambda_\text{Riesz}, \tilde{\lambda}_\text{TMLE})$ will result in similar estimators as long as
    $$
    \lambda_\text{MSE} + \lambda_\text{TMLE} = \tilde{\lambda}_{\text{MSE}} + \tilde{\lambda}_\text{TMLE}.
    $$
\end{enumerate}
We investigate these hypotheses in our numerical experiments reported in Section \ref{sec:IHDP}.

\section{Numerical Experiments}\label{sec:sim}

In this section, we apply the outcome-adapted AutoDML estimator using the neural network architecture described in Section \ref{sec:neural_nets} in four different numerical experiments using synthetic and semi-synthetic data. The details of the implementations and a link to the code are in Appendix \ref{sec:sim_study_details}.

\subsection{Outcome-adapted Dimensionality Reduction}

The first experiment is a proof of concept using a very simple data-generating process. It investigates if the outcome-adaptation and adaptive dimensionality reduction, that we have integrated into the training of our neural network, are actually capable of reducing the estimation error as expected.

We let $\beta\in\RR$ and
\begin{align*}
    W&\sim\mathcal{U}(-1,1),\\
    U\mid W=w & \sim \mathrm{Bernoulli}\left(\pi_P(w)\right),\\
    Y\mid(U,W) = (u,w) &\sim \mathcal{N}(u, 1),
\end{align*}
where $\pi_P(w) = \frac{1}{1+\exp(-\beta w)}$. We want to estimate the ATE functional from Example \ref{ex-ate}. Let $Z=h(U,W)=U$. Example \ref{ex:toy_example} shows that $Z$ satisfies the conditions of Theorem \ref{thm:delete_overadjust}, that $V_{h,P} = 4$, and that, if $\beta=0$, $V_P=4$, while, if $\beta\neq 0$,
$$
V_{P} = 2 + \frac{e^{\beta}-e^{-\beta}}{\beta}.
$$
In the best case scenario, we would like the outcome-adapted AutoDML estimator to behave as if we had oracle knowledge and simply used the representation $Z=U$. 

In the experiment we generated $2n=1000$ observations and used a simple data split where $n=500$ observations were used to estimate nuisance functions and $n=500$ observations were used to compute the final estimate. The experiment was repeated $1000$ times for different values of $\beta$ and the empirical MSE was computed for the estimators considered.

Figure \ref{fig:toy_example} shows the results of applying the outcome-adapted neural network without dimensionality reduction and with adaptive dimensionality reduction at varying levels of $\beta$. For comparison purposes, Figure \ref{fig:toy_example} includes results for an AutoDML estimator using separate neural networks to estimate the regression function and the Riesz representer. We also display $V_{h,P}$ and $V_P$ as functions of $\beta$. 

We see that the $n$-scaled MSE of the AutoDML estimator using separate neural networks is increasing in $\beta$ and approximately equals $V_P$. In contrast, the $n$-scaled MSE for the outcome-adapted neural network with adaptive dimensionality reduction does not increase with $\beta$ and is always approximately equal to $V_{h,P}=4$, that is, the estimator behaves as if we simply used the representation $Z=U$. The $n$-scaled MSE of the outcome-adapted neural network without dimensionality reduction is increasing in $\beta$ but more slowly than $V_P$. This clearly illustrates the intended effect of learning an outcome-adapted representation as well as the importance of using information bottleneck techniques. 

\begin{figure}
    \centering
    \includegraphics[width=\linewidth]{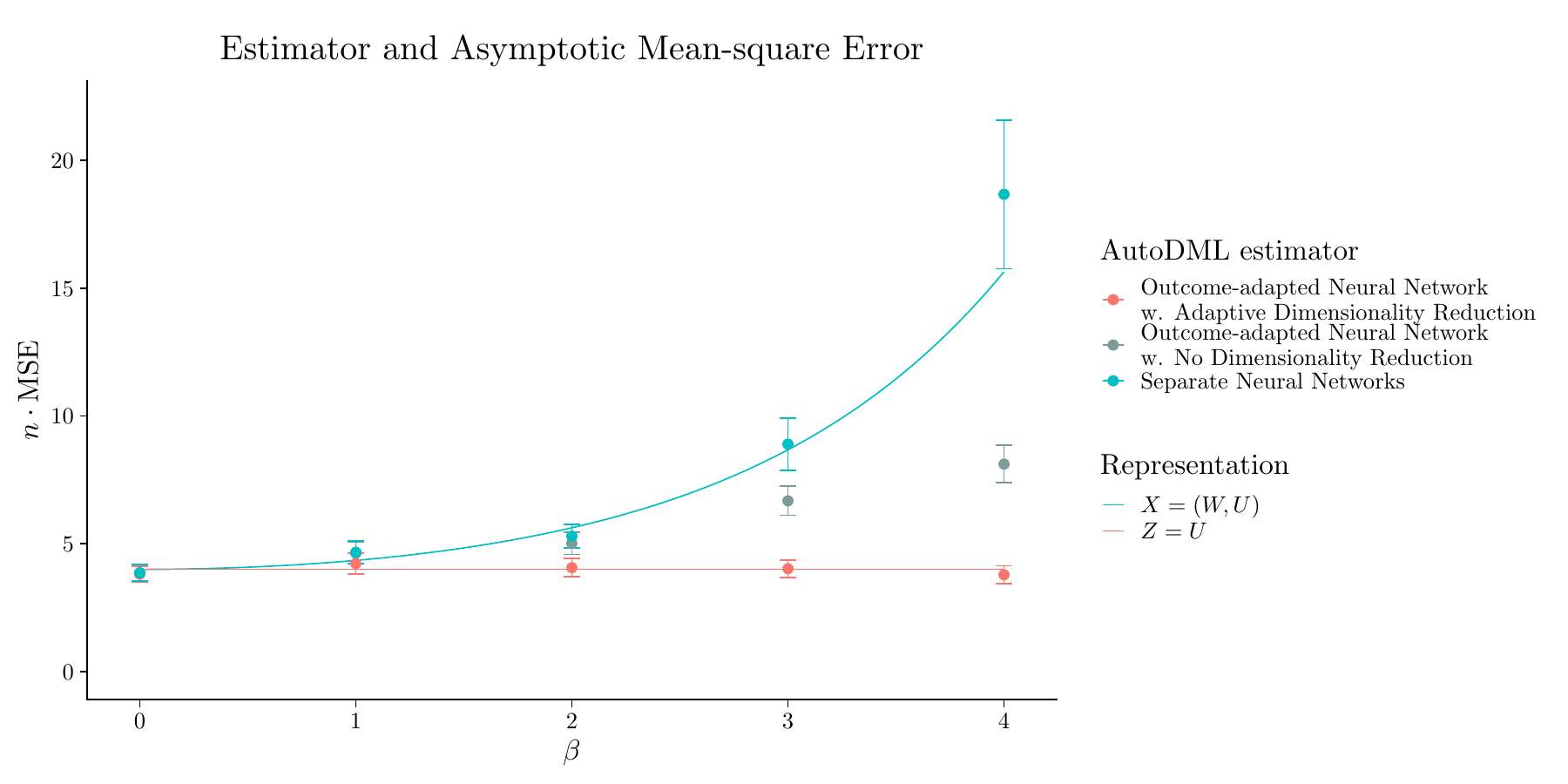}
    \caption{Scaled mean-square error (MSE) of AutoDML estimators (dots), compared to theoretical asymptotic mean square error using different representations (fully drawn lines). Error bars indicate asymptotic 95\% Gaussian confidence intervals.}
    \label{fig:toy_example}
\end{figure}

\subsection{IHDP Experiment} \label{sec:IHDP}

The second experiment benchmarks our algorithm on the IHDP dataset. The version of the IHDP dataset we use is a semi-synthetic dataset consisting of $n=737$ observations of outcomes, binary treatments and $d=25$ covariates. Outcomes represent cognitive test scores for low birth weight, premature infants in the United States between 1985 and 1988. The treatment is a binary indicator of home visits. The treatments and covariates are real data, while outcomes are simulated via the \texttt{npci} package \citep{dorie2016npci}. The author of the package has discontinued support, and we therefore procure 1000 instances of the simulated dataset from the Github repository associated to \citep{chernozhukov2022riesznetforestrieszautomaticdebiased}. This version of the IHDP benchmark is also the version used by \citet{hines2025automaticdebiasingneuralnetworks} and \citet{cai2025clearnerconstrainedlearningcausal}. Since outcomes are simulated using a known regression function, $\gamma_P$, the true ATE can be calculated as $\frac{1}{n}\sum_{i=1}^n \gamma_P(1,W_i)-\gamma_P(0,W_i)$. 

We applied the outcome-adapted neural network with and without dimensionality reduction to the 1000 instances of the IHDP dataset. We also applied an AutoDML estimator where the regression function and the Riesz representer are estimated using separate neural networks. Following \citet{chernozhukov2022riesznetforestrieszautomaticdebiased},  \citet{hines2025automaticdebiasingneuralnetworks}, and \citet{cai2025clearnerconstrainedlearningcausal}, we did not use any data splitting for this experiment. Figure \ref{fig:mae} shows the mean-absolute-error (MAE) computed from the 1000 replications. The figure also shows the MAE for RieszNet reported by \citet{chernozhukov2022riesznetforestrieszautomaticdebiased}, for MADNet reported by \citet{hines2025automaticdebiasingneuralnetworks}, and for the neural network implementation of a C-learner reported by \citet{cai2025clearnerconstrainedlearningcausal}. We see that the best performance is achieved by the outcome-adapted neural network using adaptive dimensionality reduction with MADNet achieving the second-best performance. We also see that simple dimensionality reduction yields a minor improvement over no dimensionality reduction. The AutoDML estimator using separate neural networks has the worst performance with RieszNet improving marginally on the separate neural networks.

In Section \ref{seq: RieszNet}, we argued that the performance of RieszNet is surprising if the shared representation was intended to be Riesz-adapted. Since   \citet{chernozhukov2022riesznetforestrieszautomaticdebiased} actually used a small weight of $\lambda_\text{Riesz}=0.1$  on the Riesz loss component, we hypothesized that the shared representation learned by RieszNet was actually outcome-adapted, and since the TMLE component of the network also contributes to the shared representation being outcome-adapted, this could explain its good performance.

To investigate these claims, we applied RieszNet to the IHDP datasets with varying weights on the loss function components. Firstly, we applied RieszNet with $\lambda_\text{MSE}=\lambda_\text{TMLE}=1$ fixed, as used by \citet{chernozhukov2022riesznetforestrieszautomaticdebiased}, and varied $\lambda_\text{Riesz}$. For comparison purposes, we also included an outcome-adapted and a Riesz-adapted neural network, both without dimensionality reduction. Secondly, we applied RieszNet with $\lambda_\text{Riesz}=0.1$ fixed, as used by \citet{chernozhukov2022riesznetforestrieszautomaticdebiased}, and varied $\lambda_\text{MSE}$ and $\lambda_\text{TMLE}$ such that $\lambda_\text{MSE}+\lambda_\text{TMLE}=2$ is constant. That is, given $\lambda_\text{TMLE}$, we set $\lambda_\text{MSE} = 2- \lambda_\text{TMLE}$. 

Figure \ref{fig:varying_riesz_net_weights} shows the MSEs for the different estimators. We see that with a small $\lambda_\text{Riesz}$, the performance of RieszNet is similar to the performance of the outcome-adapted neural network, while the performance of RieszNet approaches the Riesz-adapted neural network as $\lambda_\text{Riesz}$ becomes larger. When varying $\lambda_\text{TMLE}$ and keeping $\lambda_\text{MSE}+\lambda_\text{TMLE}=2$, we see that the MSE of RieszNet is relatively stable. Remarkably, we see a decent performance even with $\lambda_\text{TMLE}=2$ corresponding to removing the MSE loss component from the RieszNet loss function. These results support our hypotheses regarding the behavior of RieszNet, that is, that RieszNet actually performs best when the shared representation is outcome-adapted rather than Riesz-adapted.

\begin{figure}
    \centering
\includegraphics[width=\linewidth]{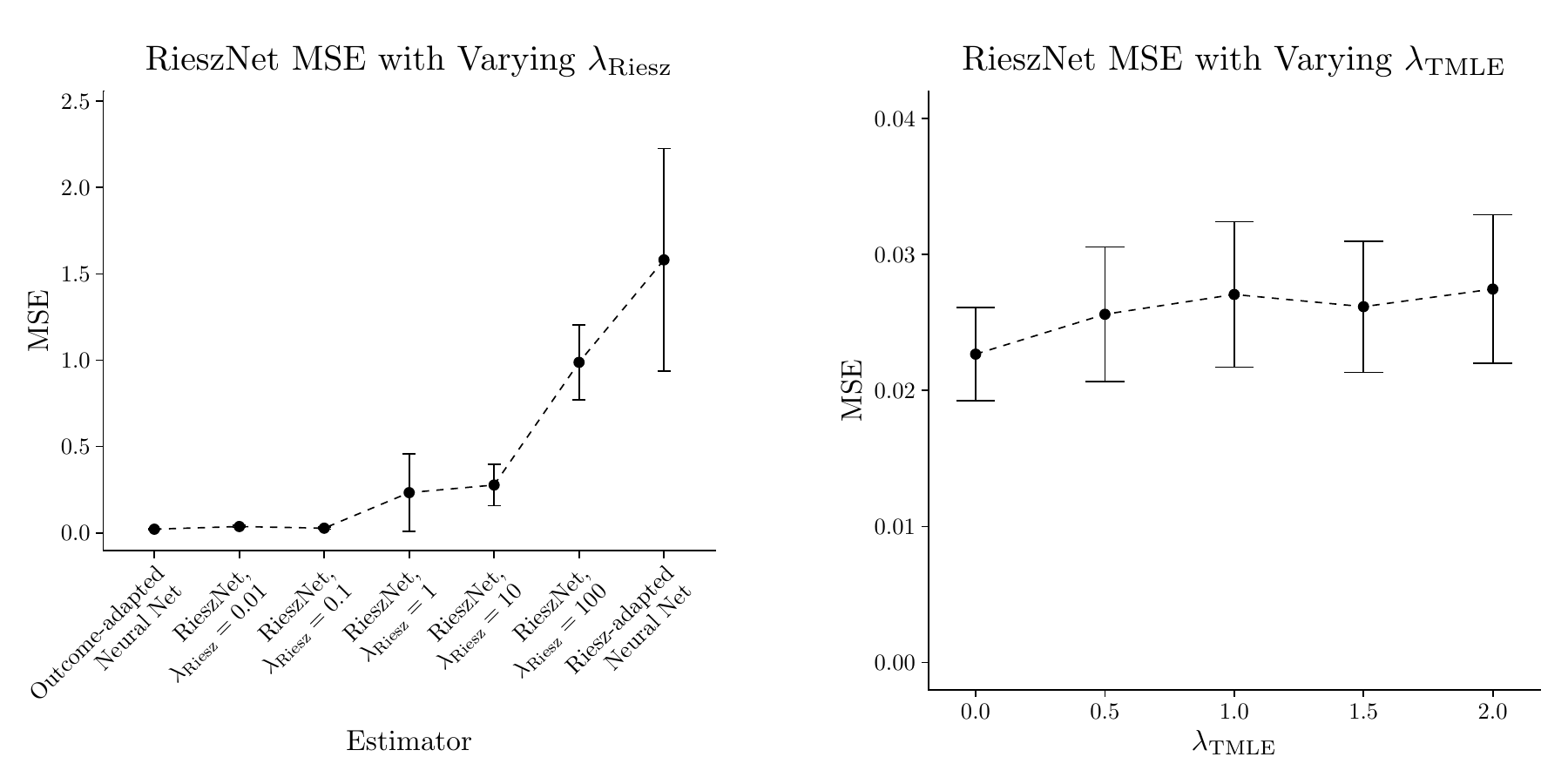}
    \caption{Mean-square error (MSE) of outcome/Riesz-adapted neural networks and RieszNet with varying weights on loss components. Error bars indicate asymptotic 95\% Gaussian confidence intervals. Dashed lines serve as visual aids only.}
    \label{fig:varying_riesz_net_weights}
\end{figure}

\subsection{Average Shift Effect on Simulated Data}\label{sec:ase_sim}

The third experiment benchmarks our algorithm for estimating the average shift effect using synthetic data. We consider a shift of size $\delta =1$ and a data-generating process given by 
\begin{align*}
W &\sim \mathcal{U}([-3,3]^4),\\ 
U &= \sin(W_2) + W_3 + (|W_1|+0.15)\epsilon_U,\\
Y &= U(W_3^2 + \exp(W_4)/2) + \epsilon_Y,
\end{align*}
where $\epsilon_U$ is standard exponential, $\epsilon_Y$ is standard normal and $(W, \varepsilon_U, \varepsilon_Y)$ are independent. The parameter of interest is the average shift effect defined in Example \ref{ex-ase}.

We generated datasets of size $2n$ with $n \in \{500, 1000, 1500, 2000\}$ and used a simple data split where $n$ observations were used to estimate nuisance functions and $n$ observations were used to compute the final estimates. All experiments were repeated 1000 times. We applied the outcome-adapted neural network with adaptive and simple dimensionality reduction as well as without dimensionality reduction. We also applied RieszNet, MADNet and an estimator that used separate neural networks for the Riesz regression and outcome regression. 

\begin{figure}
    \centering
    \includegraphics[width=\linewidth, trim={4cm 0cm 0cm 0cm}]{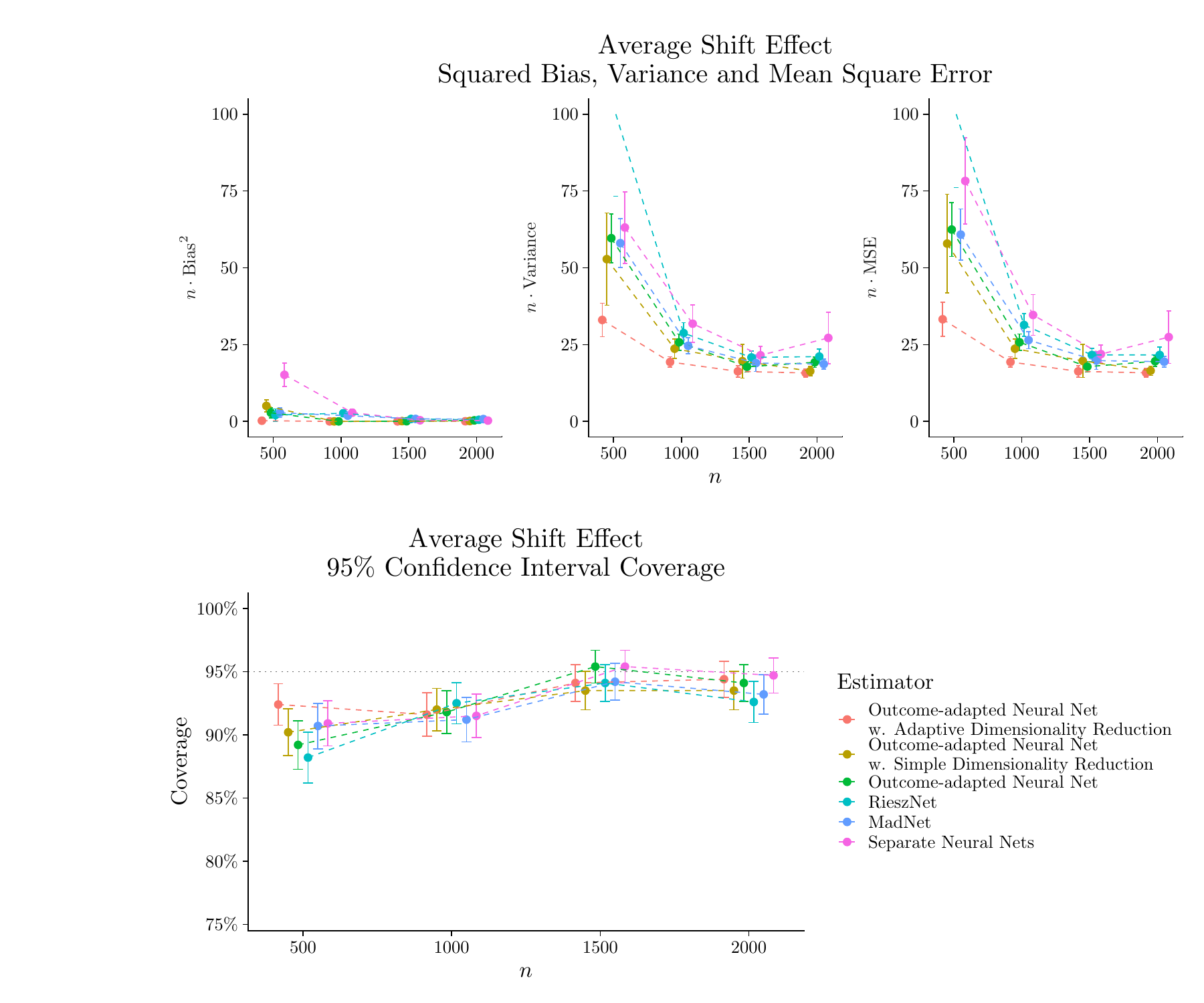}
    \caption{At the top, scaled squared bias, variance and mean square error of AutoDML estimators. At $n=500$, RieszNet has scaled variance/MSE equal to 127 which is not displayed on the plot for visibility reasons. At the bottom, coverage of AutoDML confidence intervals. Error bars indicate asymptotic 95\% Gaussian confidence intervals. Dashed lines serve as visual aids only. Points are horizontally shifted for visual clarity.}
    \label{fig:ase}
\end{figure}

Figure \ref{fig:ase} shows $n$-scaled squared bias, variance, MSE and confidence interval coverage using Gaussian confidence intervals as a function of sample size. We see that squared bias is small for all methods when $n \geq 1500$, and that squared bias in general is small compared to variance across all sample sizes, except for the estimator using separate neural networks when $n=500$, where there is a substantial amount of bias. 

We see that the outcome-adapted neural network with adaptive dimensionality reduction has the lowest variance and MSE across all sample sizes. Both of the other two outcome-adapted estimators are comparable or superior to RieszNet and MADNet in terms of variance and MSE, while the estimator using separate neural networks has the largest variance and MSE. Coverage is too low for all estimators for the small sample sizes $n =500$ and $n = 1000$, while for $n=1500$ and $n = 2000$ the coverage is close to the nominal 95\% level for all estimators. 

The results for the outcome-adapted neural network with adaptive dimensionality reduction show the benefit of using an outcome-adapted representation with an information bottleneck. As the squared bias is small for all estimators, we can attribute the decrease in MSE to the lower variance.

\subsection{Mean Missing Outcome on Simulated Data}\label{sec:mmo-sim}
The final experiment is a stress test of the estimators where we estimated the mean missing outcome from Example \ref{ex:mmo} using synthetic data. We used the data generating process originally introduced by \citet{Kang} and modified by \citet{cai2025clearnerconstrainedlearningcausal}. Specifically, with $\mathbf{I}_4$ denoting the identity matrix in $\RR^4$,  
\begin{align*}
    \tilde W &\sim\mathcal{N}(0,\mathbf{I}_4),\\
    \tilde U \mid \tilde W = \tilde w &\sim\mathrm{Bernoulli}(\tilde \pi(\tilde w)),\\
    \tilde Y \mid \tilde W = \tilde w &\sim\mathcal{N}(\tilde\gamma(\tilde w),1),
\end{align*}
and
\begin{align*}
    W_1 &= e^{\tilde W_1/2},\\
    W_2 &= \frac{\tilde W_2}{1 + e^{\tilde W_1}} + 10,\\
    W_3 &= (\tilde W_1\tilde W_3/25 + 0.6)^3,\\
    W_4 &= (\tilde W_2+\tilde W_4+20)^2,\\
    Y &= U\tilde Y, 
\end{align*}
where, with $c\geq 0$, 
\begin{align*}
    \tilde \pi(\tilde w) &= \frac{\exp (c(-\tilde w_1+0.5\tilde w_2-0.25\tilde w_3-0.1\tilde w_4))}{1 + \exp(c(-\tilde w_1+0.5\tilde w_2-0.25\tilde w_3-0. 1\tilde w_4))},\\
    \tilde \gamma(\tilde w) &= 210 + 13.7(2\tilde w_1 +\tilde w_2+\tilde w_3 +\tilde w_4).  
\end{align*}
With $c=1$, this is the data generating process developed by \citet{Kang}. The parameter $c$ was introduced by \citet{cai2025clearnerconstrainedlearningcausal}. Small $c$ values yield propensity scores that are close to $1/2$, while large $c$ values yield extreme propensity scores, that is, propensity scores close to $0$ or $1$. Our experiment used the most challenging setting considered by \citet{cai2025clearnerconstrainedlearningcausal}, $c=1.75$. This data generating process is interesting as \citet{Kang}, \citet{cai2025clearnerconstrainedlearningcausal}, and \citet{Robins_comment} all find it to be particularly challenging for doubly robust estimators like the AutoDML estimator due to the extreme propensity scores, even with $c=1$. Notice in particular that $\alpha_P(U,W) = U/\pi_P(W)$ is unbounded for any $c>0$. 

We generated datasets of size $2n$ with $n \in \{500, 1000, 1500, 2000\}$ and used a simple data split where $n$ observations were used to estimate nuisance functions and $n$ observations were used to compute the final estimates. All experiments were repeated 1000 times. We applied the outcome-adapted neural network with adaptive and simple dimensionality reduction as well as without dimensionality reduction. We also applied RieszNet, MADNet and an estimator that used separate neural networks for the Riesz regression and outcome regression. 

As \citet{cai2025clearnerconstrainedlearningcausal} find that coverage of Gaussian confidence intervals for several doubly robust estimators is too low for the data generating process with $c = 1$, we also implemented bootstrap standard-error and bootstrap percentile confidence intervals. To compute these, we resampled each of our datasets with replacement 200 times, applied our estimators, and computed the standard deviation, $\hat\sigma_\text{bootstrap}$, and bootstrapped 2.5\% and 97.5\% quantiles, $\hat{q}_{0.025}$ and $\hat{q}_{0.975}$, of the resulting 200 point estimates. The bootstrapped standard error confidence intervals were computed as $\hat \psi \pm 1.96 \cdot \hat\sigma_\text{bootstrap},$ and the bootstrapped percentile confidence intervals as $(\hat{q}_{0.025},\hat{q}_{0.975})$. For computational reasons, we only bootstrapped confidence intervals for RieszNet and the outcome-adapted neural network with adaptive dimensionality reduction.

Figure \ref{fig:mmo} shows  $n$-scaled squared bias, variance, MSE, and confidence interval coverage as a function of sample size. For visual clarity, we only show the results for the outcome-adapted neural network with adaptive dimensionality reduction and RieszNet. We defer results for the rest of the estimators as well as some additional bootstrap results to Section \ref{sec:ci-appendix} in Appendix \ref{sec:sim_study_details}. 

\begin{figure}
    \centering
    \includegraphics[width=\linewidth, trim={2.5cm 0cm 0cm 0cm}]{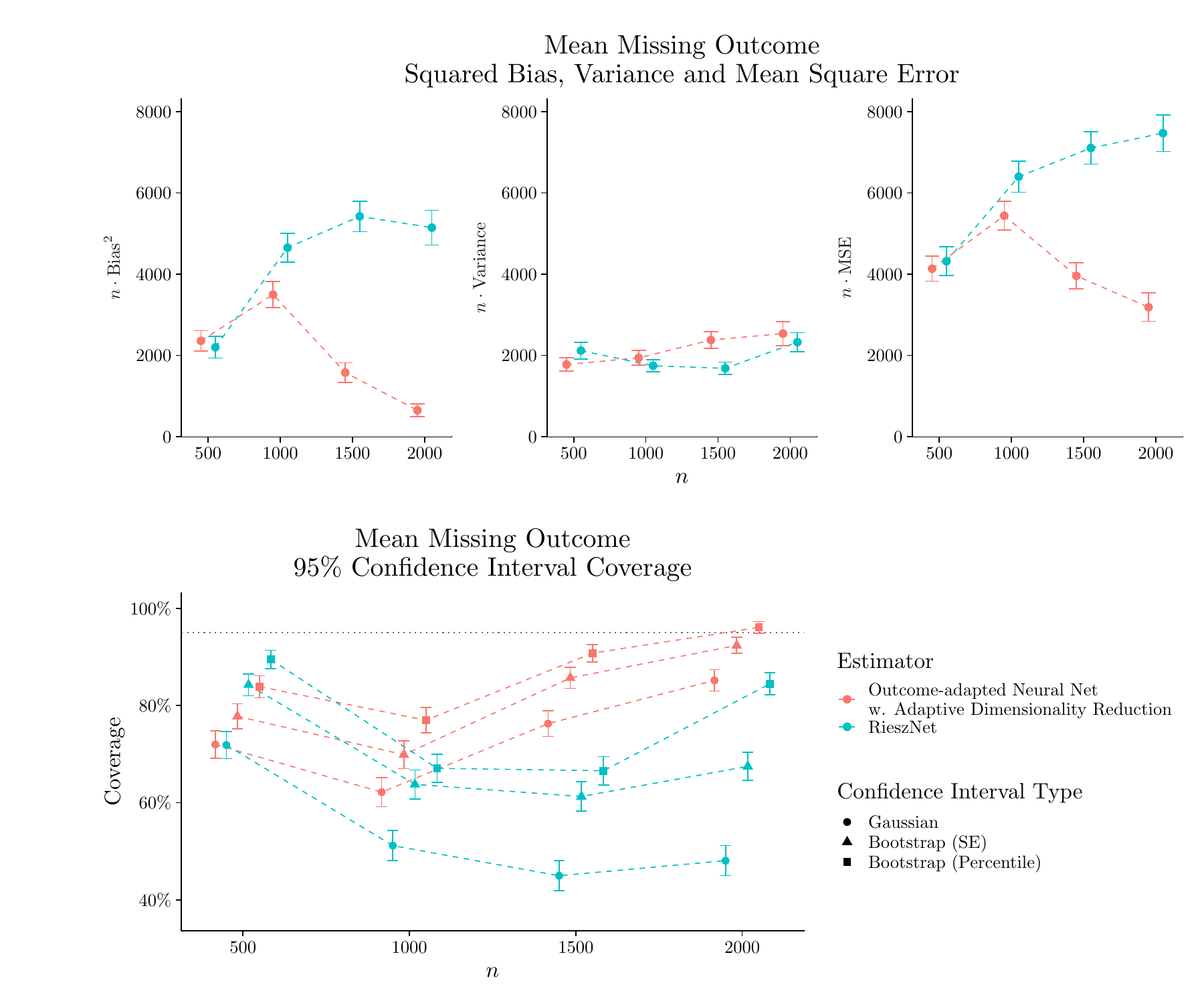}
    \caption{At the top, scaled squared bias, variance and mean square error of the outcome-adapted neural network with adaptive dimensionality reduction and RieszNet. At the bottom, coverage of Gaussian confidence intervals, as well as bootstrapped standard error and bootstrap percentile confidence intervals. Error bars indicate asymptotic 95\% Gaussian confidence intervals. Dashed lines serve as visual aids only. Points are horizontally shifted for visual clarity.}
    \label{fig:mmo}
\end{figure}

We see that RieszNet has the smallest variance across all sample sizes except for $n=500$. Conversely, the outcome-adapted neural network with adaptive dimensionality reduction has the smallest squared bias at all sample sizes except for $n=500$. The outcome-adapted neural network with adaptive dimensionality reduction has the smallest MSE at all sample sizes. At $n=1500$ and $n=2000$, in particular, the MSE of the outcome-adapted neural network is much smaller than that of RieszNet. In this case, the reduction in mean-square error of the outcome-adapted neural network with adaptive dimensionality reduction relative to RieszNet is driven by squared bias and not by variance. 

Figure \ref{fig:mmo_appendix} in Appendix \ref{sec:sim_study_details} shows that the basic AutoDML estimator using separate neural networks has MSE multiple orders of magnitude larger than what is shown in Figure \ref{fig:mmo} for RieszNet and the outcome-adapted neural net with adaptive dimensionality reduction. The large MSE of the estimator using separate neural nets is consistent with the findings of \citet{cai2025clearnerconstrainedlearningcausal}, \citet{Kang}, and \citet{robins_1995_aipw}. 

Figure \ref{fig:mmo} also shows that coverage is too low for both estimators at all sample sizes when using Gaussian confidence intervals. This is in particular the case for RieszNet at large sample sizes, while the coverage of the outcome-adapted neural network with adaptive dimensionality reduction becomes less poor as the sample size increases. Both types of bootstrap confidence intervals improve upon the Gaussian confidence intervals in terms of coverage. The percentile bootstrap confidence intervals have closest to the nominal 95\% level of coverage across all sample sizes for both RieszNet and the outcome-adapted neural network with adaptive dimensionality reduction, and nominal 95\% coverage is attained when $n=2000$ for the latter. In addition, the bootstrap standard error confidence intervals come quite close to having nominal coverage when $n=2000$ for the outcome-adapted neural network with adaptive dimensionality reduction. Aside from this, both bootstrap confidence intervals have too low coverage across sample sizes smaller than $n=2000$ for both RieszNet and the outcome-adapted neural network.

We remark that it is possible that the superior coverage of the percentile intervals is an artifact of the data generating process. Thus, we are not claiming that percentile confidence intervals are, in general, preferable to bootstrap standard error confidence intervals.

\section{Discussion}

We have shown that our outcome-adapted AutoDML estimator yields efficiency gains in terms of asymptotic variance by using a covariate representation that is predictive of outcomes and discards information about the Riesz representer. We have demonstrated that similar behavior and performance can be achieved for the related RieszNet estimator by a particular limit choice of tuning parameters, and we have also argued that this limit choice, in fact, makes RieszNet implement outcome-adaptation.  In conclusion, outcome-adaptation is the key property for achieving estimator efficiency, and the general outcome-adapted AutoDML estimator does not need tuning parameters to control a tradeoff between multiple objectives.  

We discuss here briefly some open questions and further directions of research:
\begin{itemize}
    \item \textbf{Representation optimality:} \cite{christgau2025efficientadjustmentcomplexcovariates} introduced the notion of an outcome distribution sufficient (ODS) representation. We replaced ODS by the less restrictive condition $\VP{Y \mid Z} = \VP{Y \mid X}$ in Theorem \ref{thm:delete_overadjust} to demonstrate the efficiency gain for a particular $P$. However, in this paper, we have not attempted to characterize the most efficient representation across $P \in \mathcal{P}$, which is what \cite{christgau2025efficientadjustmentcomplexcovariates} do for the estimation of the ATE. We believe it would be interesting and possible to generalize these results to all AutoDML estimators, but that it would inevitably require the $\sigma$-algebraic formulations by \cite{christgau2025efficientadjustmentcomplexcovariates}, and we leave this to future work.
    \item \textbf{Alternatives to neural network representations:} Even though our simulation studies have exclusively considered representations generated by neural networks, the outcome-adapted AutoDML estimator is not specific to neural networks. It remains an open question whether better estimators can be obtained by using other machine learning methods than neural networks for generating representations. For instance, ensemble tree-based models such as random forests and boosted trees are popular methods for regression, and a natural representation of covariates are the leaf nodes of the fitted trees as discussed by \cite{geurts2006extremely}.  
    \item \textbf{Representation error and the sure screening property:} We believe that it should be possible to deepen the theoretical analysis of the representation error, for example, by characterizing when it is asymptotically ignorable. Concretely, it is of interest to generalize the uniform sure screening property \eqref{eq:uniform-SURE} beyond a generalized linear modeling framework, and to establish conditions that ensure it holds.
    \item \textbf{Valid inference and bootstrapping:} In the average shift effect example, the asymptotic Gaussian confidence intervals were sufficient to attain nominal coverage for large enough sample sizes, while their coverage was insufficient for the adversarial mean missing outcome example. In the latter case, we demonstrated empirically that bootstrapping generally improved coverage and reached nominal coverage for sufficiently large sample sizes. There are several outstanding theoretical questions. Firstly, when the representation error is not asymptotically ignorable on the $\sqrt{n}$-scale, how should we then correctly adjust the standard error? Secondly, under which conditions is bootstrapping giving asymptotically valid confidence intervals? It is beyond the scope of this paper to answer these questions, but in light of our experiments we suggest to use bootstrapped confidence intervals in practice, even if they are 
    computationally more expensive.
\end{itemize}

\acks{AW was supported by the Independent Research Fund Denmark (DFF) through
the Sapere Aude Starting Grant (5251-00032B). AR and NRH were supported by a research grant (VIL78003) from Villum Fonden.
NRH was supported by a research grant from Novo Nordisk Fonden (NNF20OC0062897).

The authors are grateful for computing resources and technical assistance provided by the Danish Center for Climate Computing, a facility built with support of the Danish e-Infrastructure Corporation, Danish Hydrocarbon Research and Technology Centre, VILLUM Foundation, and the Niels Bohr Institute.
}

\newpage
\appendix

\section{Proofs} \label{app:proofs}
\setcounter{theorem}{0}
\setcounter{example}{0}
\setcounter{equation}{0}

\renewcommand{\thetheorem}{\thesection.\arabic{theorem}}
\renewcommand{\theexample}{\thesection.\arabic{example}}
\renewcommand{\theequation}{A\arabic{equation}}

\subsection{Proofs and Examples for Section \ref{sec:efficient-covars}}\label{seq: proofs sec:efficient-covars}

\begin{proof}[Proof of Proposition \ref{prop:rep-functional}]
    If $\eta \in\mathcal{L}^2(P_Z)$, then $\eta \circ h \in\mathcal{L}^2(P_X)$, and, since
    $$
    M_h(\eta)\circ h = M(\eta\circ h),
    $$
    it is clear that $M_h$ inherits the properties of an AutoDML operator from $M$. Furthermore, for $\eta\in L^2(P_Z)$,
    \begin{align*}
    \EP{m_h(Z,\eta)} &= \EP{m(X,\eta\circ h)} \\ &= \EP{\alpha_P(X)(\eta\circ h)(X)} \\ &= \EP{\alpha_P(X)\eta(Z)} \\ &= \EP{\alpha_{h,P}(Z)\eta(Z)}. 
    \end{align*}
    This yields the first conclusion of the proposition. 

    Now assume $\gamma_{h,P}(Z) = \gamma_{P}(X)$. Then 
    $$
    \psi_h(P) = \EP{m_h(Z,\gamma_{h,P})} = \EP{m(X,\gamma_{h,P}\circ h)} = \EP{m(X,\gamma_P)} = \psi(P). 
    $$
    If we instead assume $\alpha_{h,P}(Z) = \alpha_{P}(X)$, then 
    \begin{align*}
        \psi_h(P) &= \EP{\alpha_{h,P}(Z) \gamma_{h,P}(Z)} \\ &=
        \EP{\alpha_{h,P}(Z) Y} \\
        &=
        \EP{\alpha_{P}(X) Y} \\ &=\EP{\alpha_{P}(X) \gamma_{P}(X)}=\psi(P).
    \end{align*}
\end{proof}

\begin{proof}[Proof of Theorem \ref{thm:delete_overadjust}]
    We have that
    \begin{align*}
        &\EP{(m(X,\gamma_P) - \psi(P))\alpha_P(X)(Y-\gamma_P(X))}\\
        &= \EP{(m(X,\gamma_P) - \psi(P))\alpha_P(X)\EP{Y-\gamma_P(X) \mid X}}= 0.       
    \end{align*}
    Hence, 
    $$
    V_{P} = \EP{(m(X,\gamma_P) - \psi(P))^2} + \EP{\alpha_P(X)^2(Y-\gamma_P(X))^2}, 
    $$
    and similarly for $V_{h,P}$. Since $\gamma_P(X)=\gamma_{h,P}(Z)$, we have $m(X,\gamma_P)=m_h(Z,\gamma_{h,P})$, and, therefore,
    \begin{align*}
    V_P - V_{h,P} &= \EP{\alpha_P(X)^2(Y-\gamma_P(X))^2 - \alpha_{h,P}(Z)^2(Y-\gamma_{h,P}(Z))^2} \\
    &= \EP{(\alpha_P(X)^2-\alpha_{h,P}(Z)^2)(Y-\gamma_P(X))^2} \\
    &= \EP{(\alpha_P(X)^2-\alpha_{h,P}(Z)^2)\VP{Y\mid X}} \\
    &= \EP{(\alpha_P(X)^2-\alpha_{h,P}(Z)^2)\VP{Y\mid Z}} \\
    &= \EP{\VP{\alpha_P(X) \mid Z}\VP{Y\mid Z}}.
    \end{align*}
\end{proof}

\begin{proof}[Proof of Lemma \ref{lem:int-treat-func}]
    For $\nu\in\mathcal{M}^d$, the mapping $(u,w) \mapsto m((u,w),\nu)$ is measurable by Tonelli-Fubini. Thus, $\nu \mapsto m(\cdot,\nu)$ defines an operator $M:\mathcal{M}^d\to\mathcal{M}^d$. We check that $M$ is an AutoDML operator. For $\nu\in\mathcal{L}^2(P_X)$, by conditional Jensen's inequality and \eqref{eq:bound}, we have
    \begin{align*}
    &\EP{ \left(\int |g(\tilde u)\nu(\tilde u,W)|\dd\mu(\tilde u) \right)^2} 
    = \EP{ \left(\int \left|\frac{g(\tilde u)}{f_P(\tilde u\mid W)}\nu(\tilde u,W)\right| f_P(\tilde u\mid W)  \dd\mu(\tilde u)\right)^2} \\
    &= \EP{\EP{\left|\frac{g(U)}{f_P(U\mid W)}\nu(U,W)\right|\mid W}^2} \leq \EP{\frac{g(U)^2}{f_P(U\mid W)^2}\nu(U,W)^2}\\
    &\phantom{{}= \EP{\EP{\left|\frac{g(U)}{f_P(U\mid W)}\nu(U,W)\right|\mid W}^2} }\leq \kappa_P^2\EP{\nu(U,W)^2}. 
    \end{align*}
    From this we get for $\nu\in\mathcal{L}^2(P_X)$ that
    $$
    \int |g(\tilde u)\nu(\tilde u,w)|\dd\mu(\tilde u) < \infty
    $$
    for $P_W$-almost all $w$, showing that 
    $$
        M(\nu)(u,w) =  \int g(\tilde u)\nu(\tilde u,w) \dd\mu(\tilde u) 
    $$
    for $P_X$-almost all $(u,w)$, hence $M(\nu) \in \mathcal{L}^2(P_X)$ and $M$ is linear.
    Now the chain of inequalities above gives that 
    $$
        \EP{m((U,W),\nu)^2} \leq \EP{ \left(\int |g(\tilde u)\nu(\tilde u,W)|\dd\mu(\tilde u) \right)^2} \leq \kappa_P^2\EP{\nu(U,W)^2},
    $$
    which establishes mean-square continuity of $m$ and thus boundedness of $M$ as a 
    linear operator on $L^2(P_X)$.
    
    Since $M$ is an AutoDML operator, $\psi$ is an AutoDML functional. The claimed Riesz representer follows by calculations similar to those above, 
    \begin{align*}
    &\EP{m((U,W),\nu)} 
     = \EP{\int \frac{g(\tilde u)}{f_P(\tilde u\mid W)}\nu(\tilde u,W) f_P(\tilde u\mid W)\dd\mu(\tilde u)}\\
    &= \EP{\EP{\frac{g(U)}{f_P(U\mid W)}\nu(U,W)\mid W}} = \EP{\frac{g(U)}{f_P(U\mid W)}\nu(U,W)}. 
    \end{align*}

    Now let $h:\RR^d\to\RR^q$ be of the form $h(u,w) = (u,\tilde h(w))$ and let $Z = (U,\tilde h(W)) = (U,\tilde Z)$. With 
    $$
    m_h((u,\tilde z),\eta) =
    \int g(\tilde u) \eta(\tilde u,\tilde z) \dd\mu(\tilde u),
    $$
    whenever $\int |g(\tilde u) \eta(\tilde u,\tilde z)| \dd\mu(\tilde u)<\infty$ and $m_h((u, \tilde z), \eta) = 0$ otherwise,  we have $m_h((u,\tilde z),\eta) = m((u,w),\eta\circ h)$, and $Z$ is permissible. 
\end{proof}

\begin{proof}[Proof of Theorem \ref{thm:supp_with_precision}]
    We have
    \begin{align*}
    V_{h,P} - V_P &= \EP{\varphi_h(Z,Y,\gamma_{h,P},\alpha_{h,P},\psi(P))^2} - \EP{\varphi(X,Y,\gamma_{P},\alpha_{P},\psi(P))^2} \\
    &= \VP{R_P} + 2\EP{R_P \varphi(X,Y,\gamma_{P},\alpha_{P},\psi(P))}.
    \end{align*}
    It is thus sufficient to show
    \begin{align*}
    \EP{R_P \varphi(X,Y,\gamma_{P},\alpha_{P},\psi(P))} &= \EP{R_P m(X,\gamma_P) + R_P\alpha_P(X)(Y-\gamma_P(X)) - R_P\psi(P)} \\ &= 0.    
    \end{align*}
    Since $\alpha_P(X) = \alpha_{h,P}(Z)$,
    $$
    R_P = m_{h}(Z,\gamma_{h,P}) - m(X,\gamma_{P}) + \alpha_P(X)(\gamma_P(X)-\gamma_{h,P}(Z)). 
    $$
    In particular, $R_P$ is $\sigma(X)$-measurable, and we also see that $\EP{R_P}=0$. This gives us
    $
    \EP{R_P\psi(P)}=0,
    $
    and
    $$
    \EP{R_P\alpha_P(X)(Y-\gamma_P(X))} = \EP{R_P\alpha_P(X)\EP{Y-\gamma_P(X)\mid X}} = 0. 
    $$
    From the proof of Lemma \ref{lem:int-treat-func}, we have
    $$
    m_h((U,\tilde Z),\gamma_{h,P}) = \int g(\tilde u) \gamma_{h,P}(\tilde u,\tilde Z) \dd\mu(\tilde u).
    $$
    Therefore,  
    \begin{align*}
    \EP{\alpha_P(X)\gamma_P(X)\mid W} = \int \frac{g(\tilde u)}{f_P(\tilde u\mid W)}\gamma_P(\tilde u,W) f_P(\tilde u\mid W) \dd\mu(\tilde u)= m(X,\gamma_P),  
    \end{align*}
    and
    $$
    \EP{\alpha_P(X)\gamma_{h,P}(Z) \mid W} = \int \frac{g(\tilde u)}{f_P(\tilde u\mid W)}\gamma_{h,P}(\tilde u,\tilde Z) f_P(\tilde u\mid W)\dd\mu(\tilde u) = m_h(Z,\gamma_{h,P}). 
    $$
    Since $m_{h}(Z,\gamma_{h,P})$ and $m(X,\gamma_{P})$ are $\sigma(W)$-measurable, we get that
    $$
    \EP{R_P\mid W} = m_{h}(Z,\gamma_{h,P}) - m(X,\gamma_{P}) + m(X,\gamma_{P}) - m_h(Z,\gamma_{h,P})=0. 
    $$
    We conclude that
    $$
    \EP{R_Pm(X,\gamma_P)} =  \EP{\EP{R_P\mid W}m(X,\gamma_P)} = 0,
    $$
    which completes the proof.
\end{proof}

\begin{example}\label{ex-not permissible}
    Consider the ATE of Example \ref{ex-ate}. Assume that $W\in\{-1,1\}$ and let $h(u,w)=(2u-1)w$. Then, for $\eta:\RR\to\RR$,
    $$
    m((u,w),\eta\circ h) = (\eta\circ h)(1,w) - (\eta\circ h)(0,w) = \eta(w) - \eta(-w). 
    $$
    If we take $\eta(z)=z$, then $m((u,w),\eta\circ h) = 2w$, which is a function that cannot be written as a function of $z=h(u,w)$ since
    $$
    z = h(u,w) = 
    \begin{cases}
    1\quad & \text{ for } w=1, u=1 \text{ or } w=-1, u=0, \\ 
    -1\quad & \text{ for } w=1, u=0 \text{ or } w=-1, u=1.
    \end{cases}
    $$ 
\end{example}
    
\begin{example}\label{ex: counterexample ase.}
    Recall the average shift effect functional of Example \ref{ex-ase}. Pick $K>0$ and let 
    \begin{align*}
    W &\sim \mathcal{N}(0,1), \\
    U \mid W=w &\sim \exp(1), \\ 
    Y &= W\cdot 1_{(U > K)}. 
    \end{align*}
    We consider a shift $\delta\in(0,K)$ and the representation $Z = h(U,W) = U$. We have
    $$
    \gamma_P(U,W) =W\cdot 1_{(U > K)}, \quad \gamma_{h,P}(U) = \EP{W} 1_{(U > K)} = 0. 
    $$
    For the Riesz representer, we have 
    $$
    \alpha_P(U,W) = 1_{(U\geq \delta)}\exp(\delta )-1 = \alpha_{h,P}(U).
    $$
    Then
    $$
    V_{P} = \EP{m(X,\gamma_P)^2} = \EP{W^2} \PP{K-\delta < U \leq K } = \exp(-K) (\exp(\delta)-1),
    $$
    and 
    $$
    V_{h,P} = \EP{\alpha_{h,P}(U)^2 Y^2} = \EP{W^2 1_{(U >K)}(1_{(U \geq \delta)} \exp(\delta)-1)^2} = \exp(-K)(\exp(\delta)-1)^2.
    $$
    Thus, $V_{h,P} < V_P$ if and only if $\delta < \log(2)$.
\end{example}

\subsection{Auxiliary Probabilistic Results}\label{sec:aux_results}

This section includes two general but technical lemmas regarding uniform asymptotics, which will be used in subsequent proofs, see also Section \ref{sec-asymp} for notation. We also include a useful lemma regarding conditional expectations of functions of independent random variables.

\begin{lemma}\label{lemma: unif-cont-mapping}
Let $(X_n)_{n\in\NN}$ be a sequence of random variables on a measurable space $(\Omega,\mathbb{F})$ and let $(\mathbb{P}_P)_{P\in\mathcal{P}}$ be a family of probability measures on $(\Omega,\mathbb{F})$. 
\begin{enumerate}[label=(\arabic*)]
    \item \label{bounded a} If $X_{n} = o_\mathcal{P}(1)$ and $(A_{n,P})_{n\in\NN, P\in\mathcal{P}}$ is a collection of random variables such that $\sup_{n\in\NN, P\in\mathcal{P}} |A_{n,P}|\leq c$ for some $c>0$, then $A_{n,P} X_{n}=o_\mathcal{P}(1)$. 
    \item \label{cont mapping} If $X_{n} = a + o_\mathcal{P}(1)$ for some $a\in\RR$ and $h:\RR\to\RR$ is continuous, then $h(X_{n}) = h(a) + o_\mathcal{P}(1)$. 
    \end{enumerate}
\end{lemma}

\begin{proof}
    \begin{enumerate}
    \item[\ref{bounded a}] Let $\varepsilon>0$. Then 
    $$
    \lim_{n\to\infty}\sup_{P\in\mathcal{P}} \PP{|A_{n,P}X_{n}|>\varepsilon} \leq   \lim_{n\to\infty}  \sup_{P\in\mathcal{P}} \PP{|X_{n}|>\varepsilon/c} = 0.  
    $$

    \item[\ref{cont mapping}] Let $\varepsilon>0$. Pick $\delta>0$ such that if $|x-a|<\delta$, then $|h(x)-h(a)|<\varepsilon$. Then 
    $$
    \lim_{n\to\infty}\sup_{P\in\mathcal{P}} \PP{|h(X_{n})-h(a)|>\varepsilon}
    \leq 
    \lim_{n\to\infty}\sup_{P\in\mathcal{P}} \PP{|X_{n}-a|>\delta} = 0.  
    $$
    \end{enumerate}
\end{proof}

\begin{lemma}\label{lemma-conditional-convergence-in-prob}
    Let $(X_n)_{n\in\NN}$ be a sequence of random variables on a measurable space $(\Omega,\mathbb{F})$ and let $(\mathbb{P}_P)_{P\in\mathcal{P}}$ be a family of probability measures on $(\Omega,\mathbb{F})$. Let $(\mathcal{A}_n)_{n\in\NN}$ be a sequence of sub-$\sigma$-algebras on the background space. If 
    $$
        \EP{|X_n|\mid\mathcal{A}_n} = o_\mathcal{P}(1)
    $$  
    or
    $$ 
        \EP{X_n^2\mid\mathcal{A}_n} = o_\mathcal{P}(1),
    $$
    then 
    $$
        X_n = o_\mathcal{P}(1). 
    $$
\end{lemma}
\begin{proof}
    Since 
    $$
    \EP{|X_n|\mid\mathcal{A}_n} \leq \EP{X_n^2\mid\mathcal{A}_n}^{1/2}
    $$
    we get by Lemma \ref{lemma: unif-cont-mapping}, that $\EP{X_n^2\mid\mathcal{A}_n}= o_\mathcal{P}(1)$ implies $\EP{|X_n|\mid\mathcal{A}_n} = o_\mathcal{P}(1)$. Hence, it is sufficient to show the result for $\EP{|X_n|\mid\mathcal{A}_n}= o_\mathcal{P}(1)$.
    Let $\varepsilon > 0$. By the conditional Markov inequality,
        $$
    \PP{|X_n| > \varepsilon \mid \mathcal{A}_n} \leq \frac{\EP{|X_n| \mid \mathcal{A}_n}}{\epsilon},
    $$
    which implies that $\PP{|X_n| > \varepsilon \mid \mathcal{A}_n} = o_{\mathcal{P}}(1)$. We have for all $\delta>0$,
    \begin{align*}
        \sup_{P\in\mathcal{P}} \PP{|X_n|>\varepsilon} &= \sup_{P\in\mathcal{P}}  \EP{ \PP{|X_n| > \varepsilon \mid \mathcal{A}_n}} \\ 
        &\leq \delta +  \sup_{P\in\mathcal{P}}  \EP{ 1_{(\PP{|X_n| > \varepsilon \mid \mathcal{A}_n)> \delta} } \PP{|X_n| > \varepsilon \mid \mathcal{A}_n}} \\
        &\leq  \delta + \sup_{P \in \mathcal{P}} \PP{\PP{|X_n| > \varepsilon \mid \mathcal{A}_n}> \delta}.
    \end{align*}
    Letting $\delta\to 0$ and $n\to\infty$ gives the result. 
\end{proof}

\begin{lemma}
  \label{lemma:cond_exp_independent}
Let $(\Omega, \mathcal{F}, \mathbb{P})$ denote a probability space, $X: \Omega
\to \mathcal{X}$ and $Y: \Omega \to \mathcal{Y}$ are random variables and let
$\phi: \mathcal{X} \times \mathcal{Y} \to \mathbb{R}$ be measurable and satisfy that
$\mathbb{E}[|\phi(X, Y)|] < \infty$. Define $g: \mathcal{Y} \to \mathbb{R}$ for all
$y \in \mathcal{Y}$ by $g(y) := \mathbb{E}[\phi(X, y)]$. If $X$ and $Y$ are independent, then
\[
  \mathbb{E}[\phi(X, Y) \mid Y] = g(Y).
\]
\end{lemma}
\begin{proof}
Clearly $g(Y)$ is a measurable function of $Y$. Let $P_{XY}$ denote the
distribution of $(X, Y)$, $P_X$ the distribution of $X$ and $Y$ the distribution
of $Y$ and note that since $X$ and $Y$ are independent, we have $P_{XY} = P_X \otimes
P_Y$. For every measurable $A \subseteq \mathcal{Y}$, we have 
\begin{align*}
  \mathbb{E}[\ind_A(Y) \phi(X, Y)] &= \int_{\mathcal{X} \times \mathcal{A}} \phi(x, y) \, \mathrm{d}P_{XY}(x, y) = \int_{\mathcal{A}} \int_{\mathcal{X}} \phi(x, y) \, \mathrm{d}P_X(x) \, \mathrm{d}P_Y(y)\\
  &= \int_{\mathcal{A}} g(y) \, \mathrm{d}P_Y(y) = \mathbb{E}[\ind_{A}(Y) g(Y)].
\end{align*}
This proves the desired result.
\end{proof}

\subsection{Proofs for Section \ref{seq:sampling error}}

For $h \in \mathcal{H}$ and $P \in \mathcal{P}$, let $Z = h(X)$ and
$$
\varphi_{h,P}(Z,Y) = \varphi_{h}(Z,Y,\gamma_{h,P},\alpha_{h,P},\psi_h(P)).
$$
We also use the notation
\begin{align*}
r_\gamma(\hat Z) &= \hat\gamma_{\hat h}(\hat Z) - \gamma_{\hat h,P}(\hat Z), \\
r_\alpha(\hat Z) &= \hat\alpha_{\hat h}(\hat Z) - \alpha_{\hat h,P}(\hat Z), \\
r_\varphi(\hat Z, Y) &= \varphi_{\hat{h}}(\hat{Z},Y,\hat\gamma_{\hat{h}},\hat\alpha_{\hat{h}},\psi_{\hat{h}}(P))-\varphi_{\hat{h},P}(\hat{Z},Y).
\end{align*}
We prove Theorem \ref{thm-inferens} using two lemmas. Throughout, we assume familiarity with the results in Section \ref{sec:aux_results}.

\begin{lemma}\label{lemma-asymp_alt}
    Assume Assumptions \ref{assump-mean-square continuity} and \ref{assump-asymptotics} hold. With $(X, Y)\sim P$ independent of $D_1$, $\hat{Z} = \hat{h}(X)$ and $Z=h(X)$, the following are true, 
    \begin{align}
    \sup_{P\in\mathcal{P}}\sup_{h\in\mathcal{H}} \EP{|\varphi_{h,P}(Z,Y)|^{2+\delta}}&<\infty, \label{first claim_alt} \\
    \EP{r_\varphi(\hat Z, Y)^2\mid D_1} &= o_\mathcal{P}(1), \label{second claim_alt}  \\
     \label{third claim_alt}
    n^{1/2} \EP{r_\varphi(\hat Z, Y) \mid D_1} &= o_\mathcal{P}(1)  \\
    \inf_{P\in\mathcal{P}} \EP{\varphi_{\hat{h},P}(\hat{Z},Y)^2 \mid D_1}&>0. \label{fourth claim_alt} 
    \end{align}
\end{lemma}
\begin{proof} 
    For \eqref{first claim_alt}, by two applications of the standard inequality $|x+y|^r \leq 2^{r-1}(|x|^r+|y|^r)$ for $r \geq 1$ (provable by convexity of $t \mapsto |t|^r$) and Assumption~\ref{assump-asymptotics}~(i), we have
    \begin{align*}
    &|\varphi_{h,P}(Z,Y)|^{2+\delta} \leq 2^{1+\delta} \left( |m_h(Z,\gamma_{h,P})-\psi_{h}(P)|^{2+\delta} + |\alpha_{h,P}(Z)(Y-\gamma_{h,P}(Z))|^{2+\delta}\right)\\
    & \leq  2^{1+\delta}  |m_h(Z,\gamma_{h,P})-\psi_{h}(P)|^{2+\delta} + 2^{2+2\delta}|\alpha_{h,P}(Z)|^{2+\delta}\left(|Y|^{2+\delta} + |\gamma_{h,P}(Z)|^{2+\delta}\right)\\
    & \leq 2^{1+\delta} |m_h(Z,\gamma_{h,P})-\psi_{h}(P)|^{2+\delta} + 2^{2+2\delta}C^{2+\delta}\left(|Y|^{2+\delta} + |\gamma_{h,P}(Z)|^{2+\delta}\right).
    \end{align*}
    By taking expectations of both sides of the inequality, we obtain \eqref{first claim_alt} by Assumption~\ref{assump-asymptotics}~(i) and conditional Jensen's inequality since 
    \[
    \EP{|\gamma_{h, P}(Z)|^{2+\delta}} = \EP{|\EP{Y \mid Z}|^{2+\delta}} \leq \EP{|Y|^{2+\delta}}\leq C.
    \]

    For \eqref{second claim_alt}, we begin by writing
    \begin{equation}
    \label{eq:r_phi}
    \begin{aligned}
    r_\varphi(\hat Z, Y) &= m_{\hat{h}}(\hat{Z},\hat{\gamma}_{\hat{h}}) + \hat{\alpha}_{\hat{h}}(\hat Z)(Y-\hat{\gamma}_{\hat{h}}(\hat Z)) - m_{\hat{h}}(\hat{Z},\gamma_{\hat{h}, P}) - \alpha_{\hat{h}, P}(\hat Z)(Y-\gamma_{\hat{h}, P}(\hat Z))\\
    &= m_{\hat{h}}(\hat{Z}, \hat{\gamma}_{\hat{h}}-{\gamma}_{\hat{h},P}) - \hat\alpha_{\hat h}(\hat Z)r_\gamma(\hat Z) + r_\alpha(\hat Z)(Y -{\gamma}_{\hat{h},P}(\hat{Z})).
    \end{aligned}
    \end{equation}
    Using $(x + y + z)^2 \leq 3x^2 + 3y^2 +3z^2$ (again provable by convexity of $t \mapsto t^2$), we have 
    \begin{align*}
        \EP{r_\varphi(\hat{Z}, Y)^2 \mid D_1} &\leq 3 \EP{m_{\hat{h}}(\hat{Z}, \hat{\gamma}_{\hat{h}}-{\gamma}_{\hat{h},P})^2 \mid D_1}+ 3 \EP{\hat\alpha_{\hat h}(\hat Z)^2r_\gamma(\hat Z)^2 \mid D_1}\\
        &+ 3 \EP{r_\alpha(\hat Z)^2(Y -{\gamma}_{\hat{h},P}(\hat{Z}))^2 \mid D_1}.
    \end{align*}
    We will argue that each of the three terms in the upper bound above are $o_{\mathcal{P}}(1)$ to finish the proof that \eqref{second claim_alt} holds. For every fixed $h \in \mathcal{H}$ by Assumption~\ref{assump-mean-square continuity} and permissibility of $h$ we have for every $\eta: \mathbb{R}^q \to \mathbb{R}$ with $\EP{\eta(h(X))^2} < \infty$ that
    \begin{align*}
    \EP{m_{h}(h(X), \eta)^2} \leq \kappa \EP{ \eta(h(X))^2}.
    \end{align*}
    By Lemma~\ref{lemma:cond_exp_independent}, since $D_1$ is independent of $X$, we can use the bound for the $D_1$-measurable function $\hat{\gamma}_{\hat{h}} - \hat{\gamma}_{\hat{h}, P}$ to obtain 
    \[
        \EP{m_{\hat{h}}(\hat{Z}, \hat{\gamma}_{\hat{h}}-{\gamma}_{\hat{h},P})^2 \mid D_1} \leq \kappa \EP{ (\hat{\gamma}_{\hat{h}}(\hat{Z})-{\gamma}_{\hat{h},P}(\hat{Z}))^2 \mid D_1} = o_{\mathcal{P}}(1)
    \]
    by Assumption~\ref{assump-asymptotics}~(ii). By Assumption~\ref{assump-asymptotics}~(i) and (ii), we have
    $$
    \EP{\hat\alpha_{\hat h}(\hat Z)^2r_\gamma(\hat Z)^2 \mid D_1} \leq C^2 \EP{ r_\gamma(\hat Z)^2 \mid D_1} = o_{\mathcal{P}}(1).
    $$
    Finally, by Assumption~\ref{assump-asymptotics}~(i) and (iii), we have 
    \begin{align*}
        &\EP{r_\alpha(\hat Z)^2(Y -{\gamma}_{\hat{h},P}(\hat{Z}))^2 \mid D_1} = \EP{r_\alpha(\hat Z)^2\EP{(Y -{\gamma}_{\hat{h},P}(\hat{Z}))^2 \mid D_1, \hat Z}\mid D_1}\\
        &\leq \EP{r_\alpha(\hat Z)^2 \sup_{h \in \mathcal{H}} \EP{(Y -{\gamma}_{h,P}(h(X)))^2 \mid h(X)} \mid D_1}\\
        &\leq C \EP{r_\alpha(\hat Z)^2 \mid D_1} = o_{\mathcal{P}}(1).
    \end{align*}
    
    For \eqref{third claim_alt}, we first note that for every fixed $h \in \mathcal{H}$ by Assumption~\ref{assump-mean-square continuity} and permissibility of $h$, $m_h$ has a Riesz representer $\alpha_{h, P}$ such that for all $\eta: \mathbb{R}^q \to \mathbb{R}$ with $\EP{\eta(h(X))^2} < \infty$ we have
    \begin{align*}
    \EP{m_{h}(h(X), \eta)} = \EP{ \alpha_{h, P}(h(X)) \eta(h(X))}.
    \end{align*}
    By Lemma~\ref{lemma:cond_exp_independent}, since $D_1$ is independent of $X$, we obtain that 
    \[
        \EP{m_{h}(\hat{Z}, \eta) \mid D_1} = \EP{ \alpha_{\hat{h}, P}(\hat{Z}) \eta(\hat{Z}) \mid D_1}.
    \]
    We also note that
    $$
    \EP{r_\alpha(\hat Z) (Y -{\gamma}_{\hat{h},P}(\hat{Z})) \mid D_1} = \EP{r_\alpha(\hat Z)\EP{(Y-\gamma_{\hat{h},P}(\hat{Z}))\mid \hat Z, D_1} \mid D_1}
    = 0.
    $$
    Thus, by using \eqref{eq:r_phi}, the above observations, the Cauchy-Schwarz inequality and Assumption~\ref{assump-asymptotics}~(iv), we obtain 
    \begin{align*}
        &\left| n^{1/2} \EP{r_\varphi(\hat Z,Y) \mid D_1}\right| =
        \left| n^{1/2} \EP{ m_{\hat{h}}(\hat{Z}, \hat{\gamma}_{\hat{h}}-{\gamma}_{\hat{h},P}) - \hat\alpha_{\hat h}(\hat Z)r_\gamma(\hat Z) \mid D_1} \right| \\
        &= \left| n^{1/2} \EP{\alpha_{\hat h,P}(\hat Z)r_\gamma(\hat Z) - \hat\alpha_{\hat h}(\hat Z)r_\gamma(\hat Z) \mid D_1} \right| = \left| n^{1/2} \EP{-r_\alpha(\hat Z)r_\gamma(\hat Z)\mid D_1} \right| \\
        &\leq n^{1/2} \EP{r_\alpha(\hat Z)^2\mid D_1}^{1/2}\EP{r_\gamma(\hat Z)^2\mid D_1}^{1/2} = o_{\mathcal{P}}(1).
    \end{align*}
    
    For \eqref{fourth claim_alt}, we have for any $P \in \mathcal{P}$ by Lemma~\ref{lemma:cond_exp_independent} and Assumption~\ref{assump-asymptotics}~(i) that 
    \begin{align*}
    \EP{\varphi_{\hat h,P}(\hat{Z},Y)^2\mid D_1}
    & \geq \inf_{h \in \mathcal{H}} \EP{ \varphi_{h,P}(Z,Y)^2} \geq c > 0.
    \end{align*}
    This implies the desired result.
\end{proof}

We use the above results to show that the outcome-adapted AutoDML estimator is asymptotically linear. 

\begin{lemma}\label{lem-asymptotic linearity_alt}
    Under Assumptions \ref{assump-mean-square continuity} and \ref{assump-asymptotics}, the outcome-adapted AutoDML estimator satisfies, 
        \begin{align}
        \hat\psi - \psi_{\hat{h}}(P) &= \frac{1}{|I_2|}\sum_{i\in I_2}\varphi_{\hat{h},P}(\hat{Z}_i,Y_i) + o_\mathcal{P}(n^{-1/2}),\label{point_asymp_linear_alt}  \\
        \hat V  &= \frac{1}{|I_2|}\sum_{i\in I_2} \varphi_{\hat{h},P}(\hat{Z}_i,Y_i)^2 + o_\mathcal{P}(1). \label{var_asymp_linear_alt}
        \end{align}
\end{lemma}

\begin{proof}
For \eqref{point_asymp_linear_alt}, we have
    \begin{align*}
    \hat{\psi}- \psi_{\hat{h}}(P) 
    &= \frac{1}{|I_2|}\sum_{i\in I_2}\varphi_{\hat{h},P}(\hat{Z}_i,Y_i) + \frac{1}{|I_2|} \sum_{i \in I_2} r_\varphi(\hat Z_i, Y_i)
    \end{align*}
    Let
    \begin{align*}
    R_{1} &= |I_2|^{-1/2}\sum_{i \in I_2} \left( r_\varphi(\hat{Z}_i,Y_i) - \EP{r_\varphi(\hat{Z}_i,Y_i)\mid D_1} \right),  \\
    R_{2} &= |I_2|^{-1/2}\sum_{i \in I_2} \EP{r_\varphi(\hat{Z}_i,Y_i) \mid D_1}.
    \end{align*}
    It suffices to show that $R_1,R_2 = o_\mathcal{P}(1)$. Conditionally on $D_1$, the terms in $R_1$ are independent, identically distributed, and mean zero. By Lemma \ref{lemma-asymp_alt}, we get, with $(X,Y) \sim P$ and $\hat Z=\hat h(X)$,
    \begin{align*}
    \EP{R_1^2 \mid D_1} &= \frac{1}{|I_2|} \sum_{i \in I_2} \VP{r_\varphi(\hat{Z}_i, Y_i)\mid D_1} \leq \EP{r_\varphi(\hat{Z},Y)^2\mid D_1} = o_\mathcal{P}(1),\\
    R_2 &= |I_2|^{1/2} \EP{r_\varphi(\hat Z,Y)\mid D_1} = o_\mathcal{P}(1).  
    \end{align*}

    To see that \eqref{var_asymp_linear_alt} holds, note first that
    \begin{align*}
        \hat{V} &= \frac{1}{|I_2|} \sum_{i \in I_2} \left( \varphi_{\hat{h}, P}(\hat{Z}_i, Y_i) + r_{\varphi}(\hat{Z}_i, Y_i) -  (\hat{\psi}-\psi_{\hat{h}}(P))\right)^2\\
        &= \frac{1}{|I_2|} \sum_{i \in I_2} \left( \varphi_{\hat{h}, P}(\hat{Z}_i, Y_i) + r_{\varphi}(\hat{Z}_i, Y_i) \right)^2  + (\hat{\psi}-\psi_{\hat{h}}(P))^2\\
        &\hspace{4.025cm} - 2 (\hat{\psi}-\psi_{\hat{h}}(P)) \cdot \frac{1}{|I_2|} \sum_{i \in I_2} \left( \varphi_{\hat{h}, P}(\hat{Z}_i, Y_i) + r_{\varphi}(\hat{Z}_i, Y_i) 
        \right)\\
        &= \frac{1}{|I_2|} \sum_{i \in I_2} \left( \varphi_{\hat{h}, P}(\hat{Z}_i, Y_i) + r_{\varphi}(\hat{Z}_i, Y_i) \right)^2  - (\hat{\psi}-\psi_{\hat{h}}(P))^2\\
        &=  \frac{1}{|I_2|} \sum_{i\in I_2} \varphi_{\hat{h},P}(\hat{Z}_i,Y_i)^2 + R_3 + 2 R_4 - (\hat\psi - \psi_{\hat{h}}(P))^2,
    \end{align*}
    where 
    \begin{align*}
    R_3 =  \frac{1}{|I_2|} \sum_{i\in I_2} r_\varphi(\hat Z_i,Y_i)^2,\quad 
    R_4 =  \frac{1}{|I_2|} \sum_{i\in I_2} r_\varphi(\hat Z_i,Y_i)\varphi_{\hat h,P}(\hat Z_i,Y_i). 
    \end{align*}
    Let $(X,Y) \sim P$ independent of $D_1$, and define $\hat{Z}=\hat{h}(X)$ and $Z=h(X)$. By the triangle inequality for conditional expectations and Lemma~\ref{lemma-asymp_alt}~\eqref{second claim_alt}, we have
    \[
        \EP{|R_3| \mid D_1} \leq \EP{r_{\varphi}(\hat{Z}, Y)^2 \mid D_1} = o_{\mathcal{P}}(1).
    \]
    Similarly, by also applying conditional Cauchy--Schwarz and Lemma~\ref{lemma:cond_exp_independent}, we obtain 
    \begin{align*}
        \EP{|R_4|\mid D_1} &\leq \EP{r_\varphi(\hat Z,Y)^2 \mid D_1}^{1/2} \EP{\varphi_{\hat h,P}(\hat Z,Y)^2 \mid D_1}^{1/2}\\
        &\leq \EP{r_\varphi(\hat Z,Y)^2 \mid D_1}^{1/2} \sup_{h \in \mathcal{H}} \EP{\varphi_{h,P}(h(X),Y)^2}^{1/2}
    \end{align*}
    By Lemma~\ref{lemma-asymp_alt}~\eqref{first claim_alt} and \eqref{second claim_alt}, we conclude that $R_4 = o_{\mathcal{P}}(1)$.
    
    To establish \eqref{var_asymp_linear_alt}, it remains to show that $\hat{\psi}-\psi_{\hat{h}} = o_{\mathcal{P}}(1)$. By \eqref{point_asymp_linear_alt}, we have 
    $$
    \hat\psi - \psi_{\hat h}(P) = S_P(\hat h) + o_{\mathcal{P}}(n^{-1/2})
    $$
    where we define for $h \in \mathcal{H}$,
    $$
    S_P(h) = \frac{1}{|I_2|} \sum_{i\in I_2}\varphi_{h,P}(h(X_i),Y_i).
    $$
    By the Uniform Law of Large Numbers \citep[Lemma 19]{Shah_2020} and Lemma \ref{lemma-asymp_alt} \eqref{first claim_alt}, 
    $$
    S_P(h)=o_{\mathcal{P},\mathcal{H}}(1).
    $$
    Since $I_1,I_2$ are disjoint and $\hat h$ is estimated using the observations in $D_1$, we have for any $\varepsilon>0$, 
    \begin{align*}
    \lim_{n\to\infty}\sup_{P\in\mathcal{P}} \PP{|S_P(\hat h)|>\varepsilon} &= \lim_{n\to\infty} \sup_{P\in\mathcal{P}} \int \PP{|S_P(\hat h)|>\varepsilon\mid D_1 = d }\dd P_{D_1}(d)\\
    &\leq  \lim_{n\to\infty} \sup_{P\in\mathcal{P}} \int \sup_{h\in\mathcal{H}}\PP{|S_P(h)|>\varepsilon}\dd P_{D_1}(d)\\
    &= \lim_{n\to\infty} \sup_{P\in\mathcal{P}}\sup_{h\in\mathcal{H}} \PP{|S_P(h)|>\varepsilon} = 0. 
    \end{align*}
    This establishes $\hat\psi - \psi_{\hat h}(P) = o_\mathcal{P}(1)$. 
    \end{proof}

    \begin{proof}[Proof of Theorem \ref{thm-inferens}]
    For $h \in \mathcal{H}$, let
    $$
    \tilde{S}_P(h) = (V_{h,P}/|I_2|)^{-1/2} S_P(h).
    $$
    Combining Lemma \ref{lemma-asymp_alt} \eqref{first claim_alt} with the Uniform Central Limit Theorem \citep[Lemma 18]{Shah_2020}, we obtain
    $$
    \tilde{S}_P(h) \stackrel{d/(\mathcal{P},\mathcal{H})}{\to} \mathcal{N}(0,1).
    $$
    Since $I_1,I_2$ are disjoint and $\hat h$ is estimated using the observations in $D_1$, we have, with $\Phi$ being the distribution function of a standard Gaussian, 
    \begin{align*}
    &\lim_{n\to\infty}\sup_{P\in\mathcal{P}}\sup_{x\in\RR} \left|\PP{\tilde{S}_P(\hat h)\leq x} - \Phi(x) \right| \\ =& \lim_{n\to\infty}\sup_{P\in\mathcal{P}} \sup_{x\in\RR} \left|\int \PP{\tilde{S}_P(\hat h)\leq x\mid D_1 = d} - \Phi(x) \dd P_{D_1}(d)\right| \\
    \leq& \lim_{n\to\infty}\sup_{P\in\mathcal{P}} \sup_{x\in\RR} \int \sup_{h\in\mathcal{H}} \left| \PP{\tilde{S}_P(h)\leq x} - \Phi(x) \right| \dd P_{D_1}(d) \\
    =&
    \lim_{n\to\infty}\sup_{P\in\mathcal{P}} \sup_{h\in\mathcal{H}} \sup_{x\in\RR}  \left|\PP{\tilde{S}_P(h)\leq x} - \Phi(x) \right| = 0. 
    \end{align*}
    By combining Lemma \ref{lem-asymptotic linearity_alt} \eqref{point_asymp_linear_alt} with \citep[Lemma 20]{Shah_2020}, we get
    $$
    \left(V_{\hat h,P}/|I_2|\right)^{-1/2}(\hat\psi - \psi_{\hat{h}}(P))\stackrel{d/\mathcal{P}}{\to}\mathcal{N}(0,1).
    $$

    The final claim that  
    $$
    \hat{V} = V_{\hat{h},P} + o_{\mathcal{P}}(1)
    $$
    is shown by using the Uniform Law of Large Numbers argument to $\frac{1}{|I_2|} \sum_{i \in I_2} \varphi_{h,P}(h(X), Y)^2$ as in the proof of Lemma \ref{lem-asymptotic linearity_alt} \eqref{var_asymp_linear_alt} that was used to show $\hat{\psi} = \psi_{\hat{h}} + o_\mathcal{P}(1)$.
\end{proof}

\begin{proof}[Proof of Corollary \ref{corr-inference}]
    We have
    $$
    \left(\hat V/|I_2|\right)^{-1/2}(\hat\psi - \psi_{\hat{h}}(P)) =  \frac{\left( V_{\hat{h},P}/|I_2|\right)^{-1/2}(\hat\psi - \psi_{\hat{h}}(P))}{(\hat{V}/V_{\hat{h},P})^{1/2}}.
    $$
    Since
    $$
    \hat{V}/V_{\hat{h},P} = 1 + \frac{o_\mathcal{P}(1)}{V_{\hat{h},P}},
    $$
    we get, by Lemma \ref{lemma-asymp_alt} \eqref{fourth claim_alt},
    $$
    \left(\hat{V}/V_{\hat{h},P}\right)^{1/2} = 1 + o_{\mathcal{P}}(1).
    $$
    The first claim now follows from \citep[Lemma 20]{Shah_2020} and Theorem \ref{thm-inferens}. The second claim follows directly from the first.
\end{proof}

\subsection{Proofs for Section \ref{seq:representation error}}

\begin{proof}[Proof of Proposition \ref{thm-rep-error}]
Using mean-square continuity, we have
    \begin{align*}
    \left|\psi_{\hat{h}}(P) - \psi(P)\right|^2 &= \EP{m_{\hat{h}}(\hat{Z},\gamma_{\hat{h},P})-m(X,\gamma_P) \mid D_1} ^2 \\
    &= \EP{m(X,\gamma_{\hat{h},P}\circ \hat{h})-m(X,\gamma_P) \mid D_1} ^2 \\
    &\leq \EP{(m(X,\gamma_{\hat{h},P}\circ \hat{h})-m(X,\gamma_P))^2 \mid D_1} \\
    &\leq  \kappa_P \EP{\left(\gamma_{\hat{h},P}(\hat{Z}) - \gamma_P(X)\right)^2\mid D_1}.
    \end{align*}
\end{proof}

\begin{proof}[Proof of Corollary \ref{cor-var-select}]
Suppose $X$ is independent of $D_1$, and let $\hat{Z} = \hat{h}(X)$ where $\hat{h} : \mathbb{R}^d \to \mathbb{R}^{\hat{S}}$ depends on $D_1$ only. 
Since $\gamma_P(X) = \gamma_{h,P}(Z)$, where $Z = h_P(X)$, we
have that $\ind_{\{S(P) \subseteq \hat{S}\}}\gamma_P(X)$ is $\sigma(D_1, \hat{Z})$-measurable. Whence on the event $\{S(P) \subseteq \hat{S}\}$, we have
$$
    \gamma_{\hat{h}, P}(\hat{Z}) = \EP{Y \mid D_1, \hat{Z}} = \EP{\gamma_P(X) \mid D_1, \hat{Z}}
    = \gamma_P(X).
$$
Proposition \ref{thm-rep-error} gives that 
$$
    \left|\psi_{\hat{h}}(P) - \psi(P)\right| \leq \kappa_P^{1/2} \EP{\left(\gamma_{\hat{h},P}(\hat{Z}) - \gamma_P(X)\right)^2\mid D_1}^{1/2} \ind_{\{S(P) \not \subseteq \hat{S}\}},
$$ 
so
$$
    \mathbb{P}_P(\psi_{\hat{h}}(P) \neq \psi(P)) \leq \mathbb{P}_P(S(P) \not \subseteq \hat{S}) = 1 - \mathbb{P}_P(S(P) \subseteq \hat{S}),
$$
which shows the first part of the corollary. For the second part, take any sequence $a_n > 0$,
then for all $\varepsilon > 0$ 
$$
    \sup_{P \in \mathcal{P}} \mathbb{P}_P(\left|\psi_{\hat{h}}(P) - \psi(P)\right| > \varepsilon a_n) \leq  \sup_{P \in \mathcal{P}}  \mathbb{P}_P(\psi_{\hat{h}}(P) \neq \psi(P)) \leq 1 - \inf_{P \in \mathcal{P}} \mathbb{P}_P(S(P) \subseteq \hat{S}),
$$
showing that \eqref{eq:uniform-SURE} implies $\left|\psi_{\hat{h}}(P) - \psi(P)\right| = o_{\mathcal{P}}(a_n)$.
\end{proof}

\subsection{Proofs for Section \ref{sec:sim}}
\begin{example}\label{ex:toy_example}
    Consider the ATE functional of Example \ref{ex-ate}. Let $\beta\in\RR$ and
    \begin{align*}
        W&\sim\mathcal{U}(-1,1),\\
        U\mid W=w & \sim \mathrm{Bernoulli}\left(\pi_P(w)\right),\\
        Y\mid(U,W) = (u,w) &\sim \mathcal{N}(u, 1),
    \end{align*}
    where $\pi_P(w) = \frac{1}{1+\exp(-\beta w)}$. Let $h(u,w)=u$ and consider the representation $Z = h(U,W)=U$. Define $m_h:\RR \times \mathcal{M}^1\to\RR$ by
    $$
    m_h(z,\eta) = \eta(1) - \eta(0).
    $$
    Then 
    $$
    m_h(h(u,w),\eta) = \eta(1) - \eta(0) = m((u,w), \eta\circ h),
    $$
    and we see that $Z=h(X)$ is a permissible representation. Furthermore, 
    $$
    \gamma_{h,P}(U) = \gamma_{P}(U,W)=U, \quad \VP{Y\mid U} = \VP{Y\mid (U,W)} = 1,  
    $$
    and, hence, the conditions of Theorem \ref{thm:delete_overadjust} are satisfied. We now calculate $V_P$ and $V_{h,P}$. The conditional density of $W$ given $U=1$ is exactly $\pi_P(w)1_{(-1,1)}(w)$, and we thus have
    $$
    \EP{\frac{U}{\pi_P(W)}\mid U} = U\EP{\frac{1}{\pi_P(W)}\mid U} = U\int_{-1}^{1} \frac{\pi_P(w)}{\pi_P(w)}\dd w = 2U. 
    $$
    Similarly, it can be shown that $\EP{\frac{1-U}{1-\pi_P(W)}\mid U} = 2(1-U)$. Therefore, $\alpha_{h,P}(U)^2 = 4$, and we have
    \begin{align*}
    V_{h,P} 
    &= \EP{(m_h(U,\gamma_{h,P}) + \alpha_{h,P}(U)(Y-\gamma_{h,P}(U)) - 1)^2}\\
    &=\EP{(\alpha_{h,P}(U)(Y-\gamma_{h,P}(U)))^2} \\
    &= \EP{\alpha_{h,P}(U)^2} = 4.
    \end{align*}
    Using the moment generating function of a uniform random variable, we have, for $\beta\neq 0$, 
    \begin{align*}
    \EP{\left(\frac{U}{\pi_P(W)}\right)^2} 
    &= \EP{U\EP{\frac{1}{\pi_P(W)^2}\mid U}}\\
    &= \EP{U\int_{-1}^1 \frac{\pi_P(w)}{\pi_P(w)^2}\dd w} = 1+\frac{e^{\beta} - e^{-\beta}}{2\beta}.
    \end{align*}
    Similarly, $\EP{\left(\frac{1-U}{1-\pi_P(W)}\right)^2}=1+\frac{e^{\beta} - e^{-\beta}}{2\beta}$, and we get, for $\beta\neq 0$,
    \begin{align*}
    V_{P} 
    &= \EP{(m((U,W),\gamma_{P}) + \alpha_{P}(U,W)(Y-\gamma_{P}(U,W)) - 1)^2}\\
    &=\EP{(\alpha_{P}(U,W)(Y-\gamma_{P}(U,W)))^2} \\
    &= \EP{\alpha_{P}(U,W)^2} \\
    &= \EP{\left(\frac{U}{\pi_P(W)}\right)^2} +\EP{\left(\frac{1-U}{1-\pi_P(W)}\right)^2}
    =2+\frac{e^{\beta} - e^{-\beta}}{\beta}. 
    \end{align*}
    With $\beta = 0$, we simply have $V_{P} = 4$. 
\end{example}

\section{Relation to Semiparametric Efficiency}\label{sec: semiparametric efficiency}
\setcounter{theorem}{0}
\setcounter{example}{0}

\renewcommand{\thetheorem}{\thesection.\arabic{theorem}}
\renewcommand{\theexample}{\thesection.\arabic{example}}
The estimating function 
$$
\varphi(X,Y,\gamma,\alpha,\psi) = m(X,\gamma)+ \alpha(Y-\gamma(X)) - \psi
$$
evaluated at $\gamma=\gamma_P$, $\alpha=\alpha_P$ and $\psi=\psi(P)$ is also the efficient influence function for the AutoDML functional 
$$
\psi(P) = \EP{m(X,\gamma_P)},
$$
as shown in Theorem \ref{thm-eif} below. A heuristic argument hereof is given by \citet{newey1994}, and a similar result is shown by \citet{HirshbergWager}. For completeness, we include a proof of the result.

\subsection{Paths, Tangent sets and Influence Functions}
We refer to \citet[Chapter 25]{Vaart_1998} for background on semiparametric efficiency theory, while collecting the central definitions below. We start by defining paths that are differentiable in quadratic mean (DQM) and use these to define tangent sets. In the following definitions, we let $\mathcal{P}$ denote a generic statistical model, that is, a collection of probability measures, on some sample space $\mathcal{X}$.

\begin{definition}[DQM Path]
    Let $P \in \mathcal{P}$ and $\{P_\varepsilon:\varepsilon \in (-t,t)\} \subset \mathcal{P}$ be such that $P_0 = P$. Assume each $P_\varepsilon$ has density $p_\epsilon$ with respect to a dominating measure $\mu$. If, for some $s\in L^2(P)$,
    $$
    \int \left(\frac{\sqrt{p_\epsilon}-\sqrt{p_0}}{\varepsilon} - \tfrac{1}{2} s \sqrt{p_0} \right)^2\dd\mu \to 0 \quad \text{ as $\varepsilon\to 0$},
    $$
    we say that $\{P_\varepsilon:\varepsilon \in (-t,t)\}$ is a DQM path through $P$ with score $s$. 
\end{definition}

\begin{definition}[Tangent Set]
    Given a collection of DQM paths through some $P  \in \mathcal{P}$, we refer to the corresponding collection of scores as a tangent set at $P$. We denote tangent sets by $\mathcal{T}$.
\end{definition}

\begin{definition}[Efficient Influence Function]\label{def-eif}
    Let $P\in\mathcal{P}$ and consider a collection of DQM paths through $P$ with corresponding tangent set $\mathcal{T}$. If, for every path, $\psi(P_\varepsilon)$ is differentiable at $\varepsilon=0$, and there exists a mean-zero $\varphi_P\in L^2(P)$ such that 
    \begin{equation}\label{eq-eif-property}
        \tfrac{\partial}{\partial\varepsilon} \psi(P_\epsilon)\big|_{\varepsilon=0} = \int s \varphi_P \dd P,
    \end{equation}
    where $s$ is the score for the path, we say that $\varphi_P$ is an influence function for $\psi$ at $P$ relative to $\mathcal{T}$. If $\varphi_P\in\overline{\mathrm{span}(\mathcal{T})}$, that is, $\varphi_P$ belongs to the closure in $L^2(P)$ of the span of $\mathcal{T}$, we say that it is efficient at $P$ relative to $\mathcal{T}$.\footnote{It would be more precise to speak of influence functions relative to a collection of DQM paths as \eqref{eq-eif-property} is a property of the path. We follow the convention in semiparametric statistics and speak of efficiency relative to tangent sets.}
\end{definition}

\subsection{The Efficient Influence Function for AutoDML Functionals }
Recall the setting from Section \ref{sec:setup} where $\mathcal{P}$ is a statistical model consisting of distributions of a pair of random variables $(X,Y)$ taking values in the sample space $\RR^d \times \RR$ where $\EP{Y^2} < \infty$.

To derive the efficient influence function, we start by characterizing the collection of paths that will generate our tangent set. 
\begin{lemma}\label{lem: def of paths}
    Let $P \in \mathcal{P}$ and let $s : \RR^d \times \RR \to \RR$ be bounded and mean-zero under $P$. Let $t = \frac{1}{2}\norm{s}_{\infty}^{-1}$. For any $\epsilon \in \RR$ with $0 < |\epsilon| < t$, we can define a probability measure, $P_\epsilon$, with density 
    $$
    p_\epsilon = 1 + \epsilon s,
    $$
    with respect to $P$, and it holds that $\{P_\epsilon: {\epsilon \in (-t,t)}\}$ is a DQM path through $P$ with score $s$. Moreover, the conditional expectation of $Y$ given $X$ under $P_\epsilon$ is given by
    $$
    \gamma_{P_\epsilon}(x) = \mathbb{E}_{P_\epsilon}[Y|X=x] = \frac{\EP{Yp_\epsilon \mid X = x}}{\EP{ p_\epsilon \mid X=x}}.
    $$
\end{lemma}

\begin{proof}
    We note first that for $0 < |\epsilon| < t$, it holds that
    \begin{equation}
    \label{eq:p_eps}
    1 + \epsilon s < 1 + t \norm{s}_{\infty} \leq 2
    \quad
    \text{and}
    \quad
    1+ \epsilon s > 1 - t \norm{s}_{\infty} \geq \frac{1}{2}.
    \end{equation}
    This implies that $p_\epsilon \geq 0$ and, since $\EP{s}=0$, $p_\epsilon$ is a valid density.
    Furthermore,
    $$
    \EPepsilon{Y^2} = \EP{Y^2(1+\epsilon s)} \leq 2 \EP{Y^2} < \infty, 
    $$
    and, in particular, $\{P_\epsilon: {\epsilon \in (-t,t)}\} \subset \mathcal{P}$. We have that 
    $$
    \frac{\partial}{\partial\epsilon} \sqrt{p_\epsilon} = \frac{s}{2\sqrt{1+s\epsilon}}
    $$
    is continuous in $\epsilon$. In addition, for any $0 < |\epsilon| < t$, we have
    $$
    \int \left(\frac{\partial}{\partial \epsilon} \log(p_\epsilon)\right)^2 p_\epsilon \dd P = \int \frac{s^2}{p_\varepsilon} \dd P \leq \int 2\norm{s}_{\infty}^2 \dd P < \infty.  
    $$
    This implies continuity of 
    $$
    \epsilon \mapsto \int \left(\frac{\partial}{\partial \epsilon} \log(p_\epsilon)\right)^2 p_\epsilon \dd P
    $$ by the dominated convergence theorem. It follows from \citet[Lemma 7.6]{Vaart_1998} that $\{P_\epsilon: {\epsilon \in (-t,t)}\}$ is a DQM path with score $s$.

    For the second claim, let $A \in \sigma(X)$. Then 
    \begin{align*}        
    \mathbb{E}_{P_\epsilon}[\ind_A\gamma_{P_\epsilon}(X)] &= \EP{\ind_A\gamma_{P_\varepsilon}(X)p_\epsilon } =
    \EP{\ind_A\gamma_{P_\varepsilon}(X)\EP{p_\epsilon \mid X}} \\ 
    &= \EP{1_A\EP{Y p_\epsilon \mid X}} = \EP{\ind_Ap_\epsilon Y} = \mathbb{E}_{P_\epsilon}[\ind_AY]
    \end{align*}
    which proves the claim.
\end{proof}
\begin{theorem}\label{thm-eif}
    Let $\psi:\mathcal{P}\to\RR$ be an AutoDML functional and let $\alpha_P$ denote the induced Riesz representer. 
    Fix $P \in \mathcal{P}$ and assume that $\alpha_P$ satisfies 
    $$
    \EP{\alpha_P(X)^2(Y-\gamma_P(X))^2 } < \infty. 
    $$
    It holds that
    $$
    \varphi_P(X,Y) = m(X,\gamma_P) + \alpha_P(X) (Y - \gamma_P(X)) - \psi(P)
    $$
    is the efficient influence function at $P$ relative to the tangent set 
    $$
    \mathcal{T} = \{s \in L^2(P): \norm{s}_\infty < \infty, \EP{s} = 0\}.
    $$
\end{theorem}
\begin{proof}
    The claimed influence function is mean-zero and has second moment as
    $$
    \EP{\varphi_P(X,Y)} = \psi(P) + \EP{ \alpha_P(X) \EP{Y-\gamma_P(X)\mid X}} - \psi(P) = 0,
    $$
    and
    $$
    \EP{\varphi_P(X,Y)^2} = \EP{(m(X,\gamma_P)-\psi(P))^2} + \EP{\alpha_P(X)^2(Y-\gamma_P(X))^2} < \infty.
    $$
    Up to centering, we have $\overline{\mathrm{span}(\mathcal{T})}=L^2(P)$, and, hence, it is sufficient to verify \eqref{eq-eif-property} for a collection of paths such that the tangent set is $\mathcal{T}$. 
    
    For each $s \in \mathcal{T}$, we can, according to Lemma \ref{lem: def of paths}, define a corresponding DQM path $\{P_\epsilon: {\epsilon \in (-t,t)}\}$ with score $s$ and density $p_\epsilon = 1 + \epsilon s$ satisfying $1/2 \leq p_\epsilon \leq 2$ by \eqref{eq:p_eps}. We verify that \eqref{eq-eif-property} holds for $\varphi_P$ with respect to this path. Specifically, we wish to compute 
    \begin{align*}
    \frac{ \partial }{\partial{\epsilon}} \psi(P_\epsilon)|_{\epsilon =0} &= \frac{\partial}{\partial{\epsilon}}\EPepsilon{m(X,\gamma_{P_\epsilon})}|_{\epsilon = 0} = \frac{\partial}{\partial{\epsilon}}\EP{(1+\epsilon s)m(X,\gamma_{P_\epsilon})} |_{\epsilon = 0} \notag \\
    &= \frac{\partial}{\partial{\epsilon}}\EP{(m(X,\gamma_{P_\epsilon})} |_{\epsilon = 0} + \frac{\partial}{\partial{\epsilon}} \EP{\epsilon sm(X,\gamma_{P_\epsilon})} |_{\epsilon = 0},
    \end{align*}
    where the last equality is valid if the two derivatives in the last line exists. We claim that 
    \begin{align}
        \frac{\partial}{\partial{\epsilon}}\EP{(m(X,\gamma_{P_\epsilon})} |_{\epsilon = 0} = \EP{\alpha_P(X)(Y-\gamma_P(X))s} \label{eq:deriv1}\\
        \frac{\partial}{\partial{\epsilon}} \EP{\epsilon sm(X,\gamma_{P_\epsilon})} |_{\epsilon = 0} = \EP{m(X,\gamma_P)s} \label{eq:deriv2}.
    \end{align}
    We obtain from Lemma \ref{lem: def of paths} that 
    \begin{align*}
        &\gamma_{P_{\varepsilon}}(X) - \gamma_P(X) = \frac{\EP{Y p_\epsilon \mid X }}{\EP{p_\epsilon \mid X}} - \frac{\gamma_P(X)\EP{p_\epsilon \mid X}}{\EP{p_\epsilon \mid X}}\\
        &= \frac{\EP{Y p_\epsilon  \mid X } - \EP{\gamma_P(X)p_\epsilon \mid X}}{\EP{p_\epsilon  \mid X}} = 
        \frac{\varepsilon \EP{(Y-\gamma_P(X)) s \mid X}}{\EP{ p_\epsilon \mid X}}.
    \end{align*}    
    To see that \eqref{eq:deriv1} holds, we apply the Riesz representation of $\nu\mapsto\EP{m(X,\nu)}$ to get
    \begin{align*}
    &\frac{\EP{m(X,\gamma_{P_\epsilon})}-\EP{m(X,\gamma_{P})}}{\varepsilon}
    =\frac{\EP{\alpha_P(X)(\gamma_{P_\epsilon}(X)-\gamma_P(X))}}{\varepsilon}\\
    &=\EP{\frac{\EP{\alpha_P(X) (Y-\gamma_P(X)) s \mid X}}{ \EP{ p_\epsilon \mid X}}}.
    \end{align*}
    The integrand converges to $\EP{\alpha_P(X)(Y-\gamma_P(X))s\mid X}$ since $p_\epsilon \to 1$ as $\epsilon \to 0$ and \eqref{eq:deriv1} follows since 
    $$
    \EP{\left|\frac{\EP{\alpha_P(X) (Y-\gamma_P(X)) s \mid X}}{ \EP{ p_\epsilon \mid X}}\right|}\leq 2 \|s\|_\infty \EP{\alpha_P(X)^2(Y-\gamma_P(X))^2 }^{1/2} < \infty.
    $$
    For \eqref{eq:deriv2}, we use mean-square continuity to see that 
    \begin{align*}
    &\left|\frac{\EP{\epsilon sm(X,\gamma_{P_\epsilon})}}{\epsilon} -\EP{sm(X,\gamma_P)}\right|^2 \leq \EP{s^2 m(X, \gamma_{P_\epsilon} -\gamma_P)^2} \\ &\leq \norm{s}_\infty^2 \kappa_P \EP{(\gamma_{P_\epsilon}(X)-\gamma_P(X))^2}= \norm{s}_\infty^2 \kappa_P \epsilon^2 \EP{\left(\frac{\EP{(Y-\gamma_P(X)) s \mid X}}{ \EP{ p_\epsilon\mid X}}\right)^2}  \\
    & \leq 4\norm{s}_\infty^4 \kappa_P \epsilon^2 \EP{(Y-\gamma_P(X))^2 } \to 0 
    \end{align*}
    as $\epsilon \to 0$, where the last expectation is finite since $\EP{Y^2} < \infty$.
    
    Using that $\EP{s}=0$, we conclude that  
    $$
    \frac{\partial}{\partial{\epsilon}} \psi(P_\epsilon)|_{\epsilon =0} = \EP{(m(X,\gamma_P)+ \alpha_P (Y-\gamma_P)) s} = \EP{(m(X,\gamma_P)+ \alpha_P (Y-\gamma_P) - \psi(P)) s}.
    $$
    Hence $\varphi_P$ satisfies \eqref{eq-eif-property} and is the efficient influence function concluding the proof.
\end{proof}

\section{Details for Simulation Study}\label{sec:sim_study_details}
All code for reproducing our results is available at Github: \url{https://github.com/Asbjoern00/Outcome_Adapted_AutoDML}.

\subsection{Neural Network Implementation Details}
All neural networks are implemented using \texttt{PyTorch} \citep{paszke2019pytorchimperativestylehighperformance}. Unless otherwise specified, our neural networks have the following general configurations. The architecture is heavily inspired by \citet{chernozhukov2022riesznetforestrieszautomaticdebiased} and \citet{hines2025automaticdebiasingneuralnetworks}. 

\begin{itemize}
    \item \textbf{Architecture:} Our networks generally have 3 hidden layers with 200 neurons in the shared trunk (except when we use simple dimensionality reduction) and 2 hidden layers with 100 neurons in the branches. In the RieszNet and MADNet implementations, the Riesz branch has only one linear output layer. For the ATE experiments, we follow \citet{chernozhukov2022riesznetforestrieszautomaticdebiased} and have separate branches to return outcome predictions for the cases $U=0$ and $U=1$. 
    \item \textbf{Optimizer:} We used the \texttt{Adam} optimizer.
    \item \textbf{Early Stopping:} We used an 80/20 train/validation split and stopped training if validation loss failed to improve for 30 epochs. We trained the networks for a maximum of 1000 epochs.
    \item \textbf{Learning Rate:} We used a learning rate of $10^{-3}$ and a learning rate scheduler where learning rate was halved if validation loss failed to improve for 5 epochs.
    \item \textbf{Mini-batches:} We used mini-batches of size $64$.
    \item \textbf{Weight Decay:} We used weight decay of $10^{-3}$. In the RieszNet implementation, weight decay was not used on the $\varepsilon$ parameter as recommended by \citet{chernozhukov2022riesznetforestrieszautomaticdebiased}. 
    \item \textbf{Simple Dimensionality Reduction:} We performed 5-fold cross-validation with dimensions $\{1, 3, 10, 200\}$. We picked the smallest dimension which was within 1 standard error of the smallest cross-validation loss.
    \item \textbf{Adaptive Dimensionality Reduction:} We performed 5-fold cross-validation with $\lambda$ values $\{0,1,10,100\}$. We picked the largest $\lambda$ which is within 1 standard error of the smallest cross-validation loss. We pre-trained the model with $\lambda=0$ where we stopped training if validation loss failed to improve for 3 epochs. 
    \item \textbf{Loss Weights:} RieszNet and MADNet requires loss weights. Unless otherwise specified, we followed \citet{chernozhukov2022riesznetforestrieszautomaticdebiased} and used $\lambda_\text{Riesz}=0.1$, $\lambda_\text{MSE}=1$, and $\lambda_\text{TMLE}=1$ for RieszNet. MADNet requires choosing a weight mixing parameter, $\rho$, and a penalty parameter, $\tilde\lambda$. We followed \citet{hines2025automaticdebiasingneuralnetworks} and used $\rho=1$ and $\tilde\lambda = 5$.  
    \item \textbf{MADNet $\beta$:} MADNet requires a $\beta$ function. We use $\beta(U,W) = U$. This choice meets the requirements set by \citet{hines2025automaticdebiasingneuralnetworks}. 
\end{itemize}

We also use the following experiment-specific settings. 

\begin{itemize}
    \item \textbf{Outcome-adapted Dimensionality Reduction Experiment:} For the proof of concept experiment shown in Figure \ref{fig:toy_example}, we fixed  $\lambda=1$ for the adaptive dimensionality reduction instead of doing cross-validation. We also skipped pre-training. For large values of $\beta$, the propensity scores can be very close to 0 and 1, which lead to instability of the estimator using separate neural networks. For this reason, we capped Riesz representer estimates at $\pm 100$. 
    \item \textbf{Varying Weights on RieszNet Loss Components Experiments:} Similar to \citet{hines2025automaticdebiasingneuralnetworks}, we found that it was quite hard to replicate the performance of RieszNet found by \citet{chernozhukov2022riesznetforestrieszautomaticdebiased} due to instability of RieszNet. To remedy this instability, we increased the validation set size for the experiments, where we varied the weights of the RieszNet loss components. Specifically, we used a 60/40 train/validation split for the experiments in Figure \ref{fig:varying_riesz_net_weights}. The outcome-adapted and Riesz-adapted neural networks included in this experiment also used a 60/40 train/validation split. 
    \item \textbf{Bootstrapped Confidence Intervals}: When bootstrapping with the outcome-adapted neural network with adaptive dimensionality reduction, we use the group LASSO penalty parameter which was chosen via cross-validation when computing the point estimate. This is done for computational reasons.
\end{itemize}

\subsection{Error Bar Computation}
In all figures we quantify uncertainty with error bars. All error bars are based on asymptotic 95\% Gaussian confidence intervals of the form
$$
\hat{\xi} \pm 1.96\cdot \widehat{\text{SE}}(\hat{\xi}), \quad \xi \in \{\text{Bias}^2,\text{Var},\text{MSE},\text{Coverage}, \text{MAE}\}.
$$
For MSE, MAE and coverage we use empirical standard errors. 

Let $\hat{\psi} \in \RR^L$ denote the vector of point estimates based on $L$ independent replications and let $\bar{\psi}= L^{-1}\sum_{l=1}^L \hat{\psi}_l$ denote the mean of the point estimates. For squared bias and variance, we use standard errors given by
\begin{align*}
\widehat{\text{SE}}(\widehat{\text{Bias}^2} ) &= \frac{2\cdot |\widehat{\text{Bias}}|}{\sqrt{L}} \sqrt{ \frac{1}{L} \sum_{l=1}^L (\hat{\psi}_l - \bar{\psi})^2}, \\ 
\widehat{\text{SE}}(\widehat{\text{Var}}) &= \frac{1}{\sqrt{L}}\sqrt{ \frac{1}{L} \sum_{l=1}^L (\hat{\psi}_l - \bar{\psi})^4 - \widehat{\text{Var}}^2}.
\end{align*}

\subsection{Additional Results for Mean Missing Outcome Simulation Study} \label{sec:CaiCompare}
Recall the data generating process in Section \ref{sec:mmo-sim} introduced by \citet{cai2025clearnerconstrainedlearningcausal} and our simulation study using this data generating process also described in Section \ref{sec:mmo-sim}. 
\subsubsection{Comparison to the Simulation Study by Cai et al. (2025)}
\citet{cai2025clearnerconstrainedlearningcausal} use linear outcome models and logistic propensity score models without any sample splitting when considering $c=1.75$. Thus, the results in this simulation study are not directly comparable to those reported by \citet{cai2025clearnerconstrainedlearningcausal}. With $c=1$, \citet{cai2025clearnerconstrainedlearningcausal} also apply their estimator, the C-learner, with gradient boosting used to estimate the outcome and propensity score models and with $2$-fold cross-fitting at sample sizes $200$ and $1000$. At sample size $1000$, their method has mean absolute error $2.03$ (SE: $0.04$) while the outcome-adapted neural network with adaptive dimensionality reduction applied in the same setting has mean absolute error 1.60 (SE: 0.04).

\subsubsection{Additional Results for Mean Missing Outcome}\label{sec:ci-appendix}
Figure \ref{fig:mmo_appendix} shows $n$-scaled squared bias, variance and mean square error and Gaussian confidence interval coverage for MADNet, the outcome-adapted neural network with simple dimensionality reduction, the outcome-adapted neural network with no dimensionality reduction and the estimator using separate nets. We see that the scaled squared bias is fairly large at all sample sizes and for all estimators. In particular, there is no indication that the scaled squared bias tends to 0. 

We see that the estimator with separate neural networks and the outcome-adapted neural network with simple dimensionality reduction have mean-square error and variance that are multiple orders of magnitude larger than for MADNet and the outcome-adapted neural network with no dimensionality reduction. For the outcome-adapted neural network with simple dimensionality reduction, this is primarily due to a few replications where the error is catastrophically large, which also results in very wide error bars. The outcome-adapted neural network without dimensionality reduction consistently outperforms MADNet in terms of mean square error at all sample sizes. We see that coverage is too low for all methods considered. 

\begin{figure}
\centering
        \includegraphics[width=\linewidth, trim = 4cm 0cm 0cm 0cm]{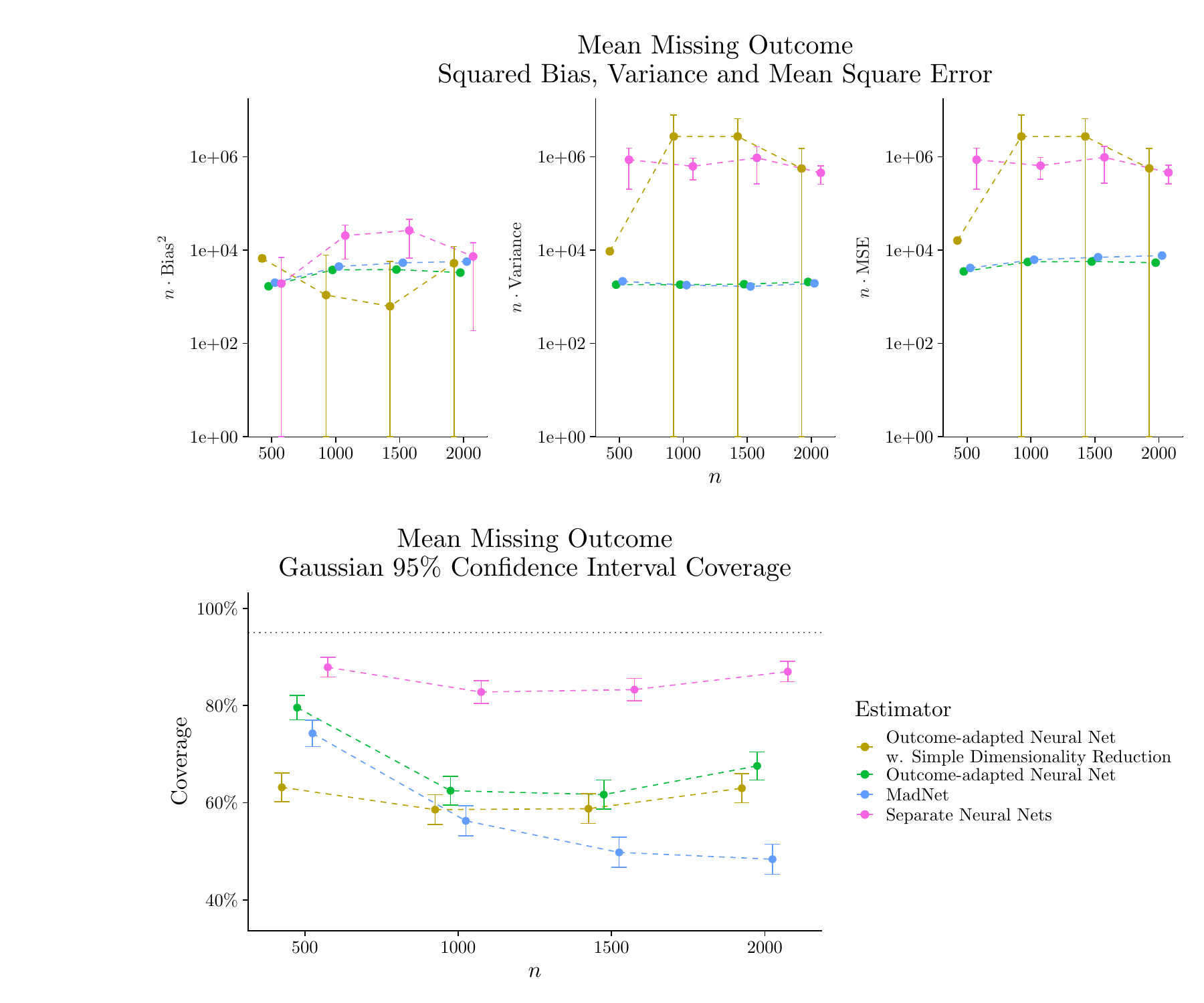}
        \caption{At the top, scaled squared bias, variance and MSE of AutoDML estimators on $\log_{10}$-scale. Error bars indicate asymptotic 95\% Gaussian confidence intervals. Dashed lines serve as visual aids only. Points are horizontally shifted for visual clarity.}
        \label{fig:mmo_appendix}
    \label{fig:bootstrap_mmo_combined}
\end{figure}

We further compare the different types of confidence intervals. In Figure \ref{fig:bootstrap}, we display the 100 first confidence intervals for RieszNet and the outcome-adapted neural network with adaptive dimensionality reduction when $n = 2000$. We see that there is a substantial amount of bias for both estimators and that the bias is largest for RieszNet. The bias results in too low coverage for the Gaussian and bootstrap standard error intervals which are symmetric around the point estimate. The percentile intervals do not suffer from this issue to the same extent, yielding better coverage for both estimators. However, the percentile confidence intervals have the undesirable property that they sometimes do not cover the point estimate, probably due to the relatively small number of bootstrap resamples.

\begin{figure}
    \centering
        \centering
        \includegraphics[width=\linewidth]{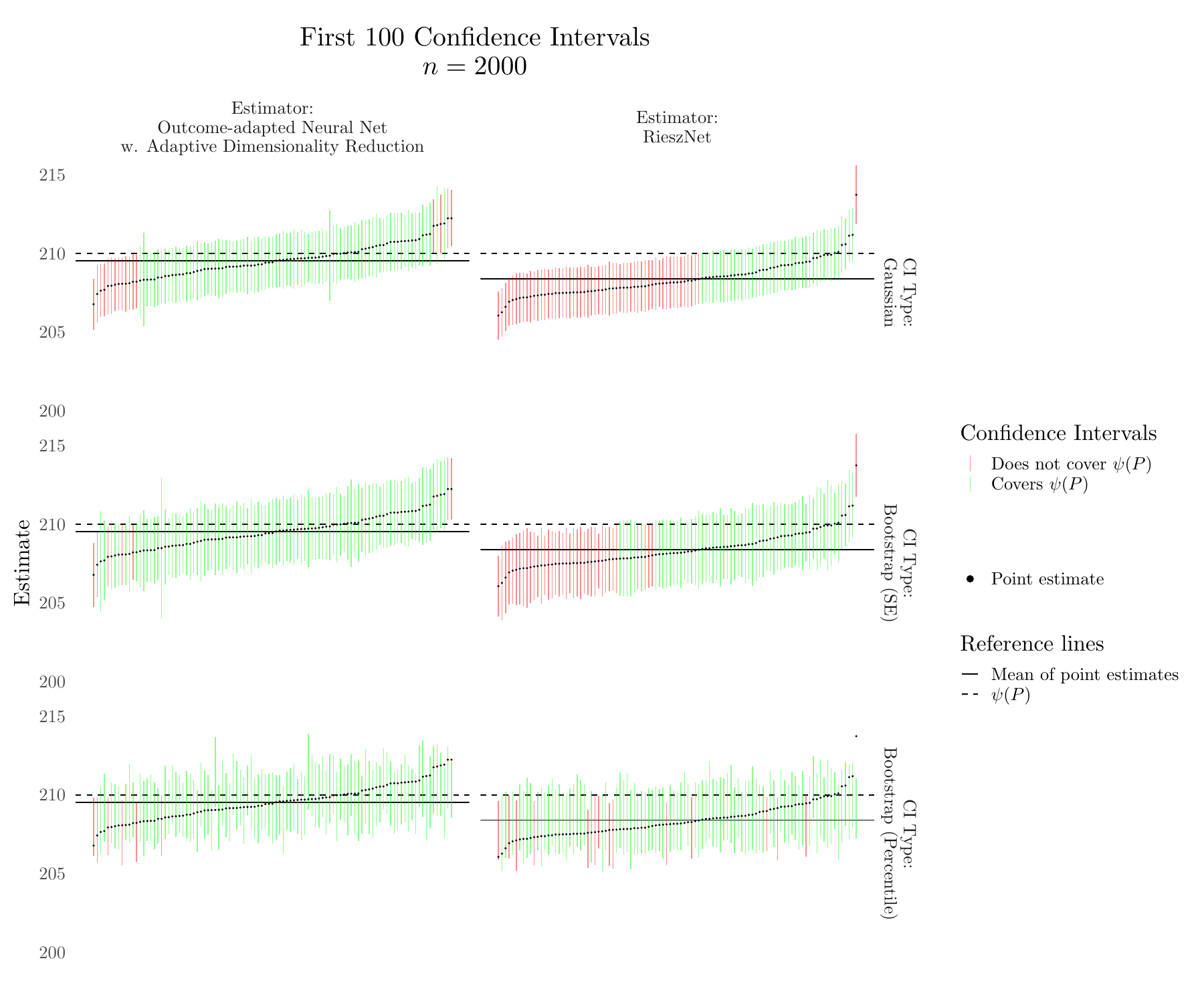}
        \caption{The first 100 Gaussian, bootstrap percentile and bootstrap standard error confidence intervals for RieszNet and the outcome-adapted neural network with adaptive dimensionality reduction, along with point estimates. Confidence intervals are sorted by point estimate.}
        \label{fig:bootstrap}
\end{figure}

\vspace{0.2in}

\bibliography{references}

@misc{chernozhukov_24_via_riesz_regression,
      title={Automatic debiased machine learning via {Riesz} regression}, 
      author={V. Chernozhukov and W. K. Newey and V. Quintas-Martinez and V. Syrgkanis},
      year={2024},
      note= {\textit{arXiv:2104.14737}},
}

@misc{singh2024kernelridgerieszrepresenters,
      title={Kernel ridge {Riesz} representers: Generalization, mis-specification, and the counterfactual effective dimension}, 
      author={R. Singh},
      year={2024},
      note = {\textit{arXiv:2102.11076}},
      eprint={2102.11076},
      archivePrefix={arXiv},
      primaryClass={stat.ML},
}

@misc{lee2025rieszboostgradientboostingriesz,
      title={RieszBoost: Gradient boosting for {Riesz} regression}, 
      author={K. J. Lee and A. Schuler},
      year={2025},
      eprint={2501.04871},
      note = {\textit{arXiv:2501.04871}},
      archivePrefix={arXiv},
      primaryClass={stat.ML},
}

@misc{christgau2025efficientadjustmentcomplexcovariates,
      title={Efficient adjustment for complex covariates: Gaining efficiency with {DOPE}}, 
      author={A. M. Christgau and A. R. Lundborg and N. R. Hansen},
      year={2025},
      note = {\textit{arXiv:2402.12980}},
      eprint={2402.12980},
      archivePrefix={arXiv},
      primaryClass={math.ST},
}

@misc{cai2025clearnerconstrainedlearningcausal,
      title={C-Learner: Constrained Learning for Causal Inference}, 
      author={T. T. Cai and Y. Fonseca and K. Hou and H. Namkoong},
      note = {\textit{arXiv:2405.09493}},
      year={2025},
      eprint={2405.09493},
      archivePrefix={arXiv},
      primaryClass={stat.ML},
}

@misc{dorie2016npci,
  author = {V. Dorie},
  title  = {Non-parametrics for Causal Inference},
  year   = {2016},
  url    = {https://github.com/vdorie/npci}
}

@article{HirshbergWager,
author = {D. A. Hirshberg and S. Wager},
title = {{Augmented minimax linear estimation}},
volume = {49},
journal = {The Annals of Statistics},
number = {6},
publisher = {Institute of Mathematical Statistics},
pages = {3206 -- 3227},
year = {2021},
}

@article{shi2019adapting,
 author = {Shi, C. and Blei, D. and Veitch, V.},
 journal = {Advances in Neural Information Processing Systems},
 publisher = {Curran Associates, Inc.},
 title = {Adapting Neural Networks for the Estimation of Treatment Effects},
 volume = {32},
 year = {2019}
}

@article{chernozhukov2018doubledebiasedmachinelearningtreatment,
    author = {Chernozhukov, V. and Chetverikov, D. and Demirer, M. and Duflo, E. and Hansen, C. and Newey, W. and Robins, J.},
    title = {Double/debiased machine learning for treatment and structural parameters},
    journal = {The Econometrics Journal},
    volume = {21},
    number = {1},
    pages = {C1-C68},
    year = {2018},
}

@article{Chernozhukov_2022_auto_dml_lasso,
author = {Chernozhukov, V. and Newey, W. K. and Singh, R.},
title = {Automatic Debiased Machine Learning of Causal and Structural Effects},
journal = {Econometrica},
volume = {90},
number = {3},
pages = {967-1027},
year = {2022}
}

@article{robins_1995_aipw,
 author = {J. M. Robins and A. Rotnitzky},
 journal = {Journal of the American Statistical Association},
 number = {429},
 pages = {122--129},
 publisher = {[American Statistical Association, Taylor & Francis, Ltd.]},
 title = {Semiparametric Efficiency in Multivariate Regression Models with Missing Data},
 volume = {90},
 year = {1995}
}

@article{rotnitzky2019efficientadjustmentsetspopulation,
  author  = {A. Rotnitzky and E. Smucler},
  title   = {Efficient Adjustment Sets for Population Average Causal Treatment Effect Estimation in Graphical Models},
  journal = {Journal of Machine Learning Research},
  year    = {2020},
  volume  = {21},
  number  = {188},
  pages   = {1--86},
}

@article{Henckel_2022,
    author = {Henckel, L. and Perković, E. and Maathuis, M. H.},
    title = {Graphical Criteria for Efficient Total Effect Estimation Via Adjustment in Causal Linear Models},
    journal = {Journal of the Royal Statistical Society Series B: Statistical Methodology},
    volume = {84},
    number = {2},
    pages = {579-599},
    year = {2022},
}

@article{OAL,
author = {Shortreed, S. M. and Ertefaie, A.},
title = {Outcome-adaptive lasso: Variable selection for causal inference},
journal = {Biometrics},
volume = {73},
number = {4},
pages = {1111-1122},
year = {2017}
}

@article{Shah_2020,
   title={The hardness of conditional independence testing and the generalised covariance measure},
   volume={48},
   number={3},
   journal={The Annals of Statistics},
   publisher={Institute of Mathematical Statistics},
   author={Shah, R. D. and Peters, J.},
   year={2020}}

@article{newey1994,
 author = {W. K. Newey},
 journal = {Econometrica},
 number = {6},
 pages = {1349--1382},
 publisher = {[Wiley, Econometric Society]},
 title = {The Asymptotic Variance of Semiparametric Estimators},
 volume = {62},
 year = {1994}
}

@article{hines2025automaticdebiasingneuralnetworks,
  title = 	 {Automatic debiasing of neural networks via moment-constrained learning},
  author =       {Hines, C. L. and Hines, O. J.},
  pages = 	 {390--405},
  year = 	 {2025},
  volume = 	 {275},
  journal = 	 {Proceedings of Machine Learning Research},
  publisher =    {PMLR},
}

@article{chernozhukov2022riesznetforestrieszautomaticdebiased,
  title = 	 {{R}iesz{N}et and {F}orest{R}iesz: Automatic Debiased Machine Learning with Neural Nets and Random Forests},
  author={V. Chernozhukov and W. K. Newey and V. Quintas-Martinez and V. Syrgkanis},
  pages = 	 {3901--3914},
  year = 	 {2022},
  volume = 	 {162},
  journal = 	 {Proceedings of Machine Learning Research},
  publisher =    {PMLR},
}

@article{Kang,
author = {J. D. Y. Kang and J. L. Schafer},
title = {{Demystifying Double Robustness: A Comparison of Alternative Strategies for Estimating a Population Mean from Incomplete Data}},
volume = {22},
journal = {Statistical Science},
number = {4},
publisher = {Institute of Mathematical Statistics},
pages = {523 -- 539},
year = {2007},
}

@article{Robins_comment,
author = {J. Robins and M. Sued and Q. Lei-Gomez and A. Rotnitzky},
title = {Comment: {P}erformance of Double-Robust Estimators When “Inverse Probability” Weights Are Highly Variable},
volume = {22},
journal = {Statistical Science},
number = {4},
publisher = {Institute of Mathematical Statistics},
pages = {544 -- 559},
year = {2007}
}

@article{paszke2019pytorchimperativestylehighperformance,
title = {PyTorch: An Imperative Style, High-Performance Deep Learning Library},
author = {Paszke, A. and Gross, S. and Massa, F. and Lerer, A. and Bradbury, J. and Chanan, G. and Killeen, T. and Lin, Z. and Gimelshein, N. and Antiga, L. and Desmaison, A. and Kopf, A. and Yang, E. and DeVito, Z. and Raison, M. and Tejani, A. and Chilamkurthy, S. and Steiner, B. and Fang, L. and Bai, J. and Chintala, S.},
volume = {32},
journal = {Advances in Neural Information Processing Systems},
pages = {8024--8035},
year = {2019},
publisher = {Curran Associates, Inc.},
}

@article{geurts2006extremely,
  title={Extremely randomized trees},
  author={Geurts, Pierre and Ernst, Damien and Wehenkel, Louis},
  journal={Machine Learning},
  volume={63},
  number={1},
  pages={3--42},
  year={2006},
  publisher={Springer}
}

@book{Vaart_1998,
title={Asymptotic statistics},
publisher={Cambridge University Press},
author={van der Vaart, A. W.},
year={1998},
}

@book{vanderLaanRose2011,
  author    = {M. J. van der Laan and S. Rose},
  title     = {Targeted learning: Causal inference for observational and experimental data},
  year      = {2011},
  publisher = {Springer},
}

@article{fanSureIndependenceScreening2008,
  title = {Sure Independence Screening for Ultrahigh Dimensional Feature Space},
  author = {Fan, J. and Lv, J.},
  year = {2008},
  journal = {Journal of the Royal Statistical Society Series B: Statistical Methodology},
  volume = {70},
  number = {5},
  pages = {849--911}
}

@article{wangHighDimensionalOrdinary2016,
    title = {High Dimensional Ordinary Least Squares Projection for Screening Variables},
    volume = {78},
    number = {3},
    journal = {Journal of the Royal Statistical Society Series B: Statistical Methodology},
    author = {Wang, X. and Leng, C.},
    year = {2016},
    pages = {589--611}
}

\end{document}